\documentclass[a4paper,%
]{easychair}
\usepackage{booktabs}   
\usepackage{subcaption} 

\usepackage{tikzsymbols}
\usetikzlibrary{arrows,shapes}
\usepackage{pict2e}

\usepackage{amsmath}
\usepackage{amssymb}
\usepackage{latexsym}
\usepackage{proof}
\usepackage{microtype}

\usepackage{bussproofs}
\usepackage{amssymb}
\usepackage{amsmath}
\usepackage{mathrsfs}
\usepackage{upgreek}
 \usepackage{graphicx}

\usepackage{hyperref}

\usepackage{algpseudocode}

\newtheorem{definition}{Definition}
       \newtheorem{example}[definition]{Example}
       
       \newtheorem{theorem}[definition]{Theorem}

       \newtheorem{proposition}[definition]{Proposition}
       
       \newtheorem{remark}[definition]{Remark}

       \newtheorem{corollary}[definition]{Corollary}
\usepackage{tikz}
\usetikzlibrary{arrows, shapes, decorations.pathmorphing}

\newdimen\CdotAxis
\newcommand*{\CdotAux}[3]{%
  {%
    \settoheight\CdotAxis{$#2\vcenter{}$}%
    \sbox0{%
      \raisebox\CdotAxis{%
        \scalebox{#1}{%
          \raisebox{-1.1pt}{%
            $\mathsurround=0pt #2#3$%
          }%
        }%
      }%
    }%
    \dp0=0pt %
    \sbox2{$#2\bullet$}%
    \ifdim\ht2<\ht0 %
      \ht0=\ht2 %
    \fi
    \sbox2{$\mathsurround=0pt #2#3$}%
    \hbox to \wd2{\hss\usebox{0}\hss}%
  }%
}

\newcommand{\substitution} [1]         { {\vec{#1}} }

\newcommand{\redn}              {\rightsquigarrow}
\newcommand{\red}              {\;\mathsf{r}\;}

\newcommand{\proj}                     { {\mathsf{\uppi}} }
\newcommand{\emp}[1]    {{\mathsf{P}_{#1}}}

\newcommand{\cruno}           {{\textbf{(CR1)}}}
\newcommand{\crdue}           {{\textbf{(CR2)}}}
\newcommand{\crtre}           {{\textbf{(CR3)}}}

\newcommand{\sn}{\mathsf{SN}}
\newcommand{\snn}{\mathsf{SN}^{\star}}

\newcommand{\pair}[2]{\langle #1,#2\rangle}

\newcommand{\IMPL}{\rightarrow}

\newcommand{\impl}{\rightarrow}
\newcommand{\et}{\wedge}
\newcommand{\vel}{\vee}
\newcommand{\fal}{\bot}
\newcommand{\ver}{\top}
\newcommand{\non}{\neg}

\newcommand{\ax}{\mathcal{A}}
\newcommand{\axx}{\mathsf{Ax}}

\newcommand{\contr}{ \textit{contr}\hspace{1pt}}

\newcommand{\lama}{\lambda_{\p}}

\newcommand{\inappendix}[1]{}

\newcommand{\lam}{\lambda}
\newcommand{\lan}{\langle}
\newcommand{\ran}{\rangle}
\newcommand{\efq}[2]{#2 \, \mathsf{efq}_{#1 }}

\newcommand{\p}{\parallel}
\newcommand{\chosen}[1]{\underline{ #1 }}
\newcommand{\send}[1]{\overline{#1}\,}
\newcommand{\get}[1]{\underline{#1}\,}

\newcommand{\pp}[2] {  {}_{#1} \hspace{-1pt}  ( #2 ) }
\newcommand{\ppp}[3] {{}^{#2}_{\; #1}\hspace{-1pt}( #3 ) }
\newcommand{\po}[1]{{}_{#1} \hspace{-1pt}  \big(}
\newcommand{\pc}{\big)}

\newcommand{\IL}{\mathbf{IL}}
\newcommand{\NJ}{\mathbf{NJ}}

\newtheorem{procedure}{Procedure}

\def\NJ{{\bf NJ}}

\def\Deduce#1{\hbox{$\hphantom{#1}$\kern\inferLabelSkip\DeduceSym\kern\inferLab
elSkip$#1$}}

\newcommand\shortdots{\!\makebox[1em][c]{.\hfil.\hfil.}}

\newcounter{inlinelabel}
\let\theinlinelabelbak\theinlinelabel
\renewcommand{\theinlinelabel}{(\theinlinelabelbak)}
\newcommand{\inlinelabel}[1]{
        \refstepcounter{inlinelabel}\label{#1}%
        \theinlinelabel
}

\title{A typed parallel $\lambda$-calculus via 1-depth intermediate proofs\thanks{Funded by the FWF grants Y544-N23 and P32080-N31, and project IBS (ANR-18-CE27-0012-01).}}

\author{Federico Aschieri\inst{1} \and Agata Ciabattoni\inst{1} \and Francesco A.\ Genco\inst{2}}
\institute{TU Wien, Vienna, Austria \\
\email{federico.aschieri@tuwien.ac.at, agata@logic.at}\\
 \and 
IHPST, Universit\'e Paris 1 and CNRS, Paris, France \\
\email{frgenco@gmail.com} 
}

\authorrunning{Aschieri, Ciabattoni and Genco}
\titlerunning{A typed parallel $\lambda$-calculus via 1-depth intermediate proofs}

\begin{document}

\maketitle

\begin{abstract}
%
 We introduce a Curry--Howard correspondence
 for a large class of intermediate logics characterized
 by intuitionistic proofs with non-nested applications of 
 rules for classical disjunctive tautologies (1-depth intermediate proofs).
 The resulting calculus, we call it  $\lama$, is a strongly normalizing parallel extension of the
 simply typed $\lambda$-calculus.
 Although simple, the $\lama$ reduction rules can model arbitrary
process network topologies, and encode interesting parallel programs ranging
from numeric computation to algorithms on graphs.
\end{abstract}

\section{Introduction}\label{sec:intro}
The fundamental connection between  logic and computation,
known as Curry--Howard correspondence, 
relates logics and computational systems.
Originally introduced  
for intuitionistic logic $\IL$ and simply typed $\lambda$-calculus,
it has been extended to many different logics (classical~\cite{Parigot92},
linear~\cite{CairesPfenning,AG2020}, modal~\cite{MCHP}\dots) and notions of 
computation, see, e.g., \cite{Wadler} for an overview.

A recent addition to them is the discovery in \cite{ACGTCS} of the
connection between propositional logics intermediate between
classical logic and $\IL$, 
and concurrent extensions of the simply typed $\lambda$-calculus.
More precisely, the considered logics extend $\IL$ with the classical disjunctive tautologies
interpreted as synchronization schemata in~\cite{DanosKrivine},
 i.e.~axiom schemata of the form $(F_1 \impl G_1)\vel \ldots \vel (F_m \impl G_m) $.
This general result was preceeded by                 
Curry--Howard correspondences for the
 {\em particular cases} of G\"odel-Dummett logic
($\IL$ with LIN $= (A \to B) \vee (B \to A)$)
and classical logic ($\IL$ with EM $= A \vee \neg A$) 
in~\cite{ACGlics2017,ACG2018}.
All these logics possess cut-free hypersequent calculi --
a generalization of sequent calculi consisting of parallel compositions of sequents that can ``communicate''; our results confirmed Avron's conjecture \cite{Avron91} of the connection between (intermediate
logics characterized by cut-free) hypersequent calculi and concurrency. 
Although the resulting typed concurrent $\lambda$-calculi provided an adequate computational interpretation
of these logics, they are quite sophisticated and not easy to use as programming languages. 
Their main reductions are:
{\em intuitionistic reductions} -- which are the usual computational rules 
for the simply typed $\lambda$-calculus~\cite{Girard}, 
{\em permutation reductions} -- needed to prove weak normalization, and the reductions implementing communications among
simply typed $\lambda$-terms; these are divided into {\em basic cross reductions} -- simple reductions for the natural deduction version of the characteristic hypersequent rules, and {\em full cross reductions} -- needed for the
subformula property -- which enable the transmission of messages that depend on their computational environment by using 
code mobility concepts~\cite{codemobil} such as that of closure~\cite{landin}. 


The aim of this work is to present a version of the calculi in \cite{ACGTCS,ACGlics2017,ACG2018} suitable for programming.
We show that a simplified version of these calculi is expressive enough to model arbitrary process network topologies and to encode interesting
parallel algorithms taken from \cite{Eden}.
Inspired by \cite{Marcello,MFG13}
we base our Curry--Howard correspondence on 
1-depth\footnote{
\cite{Marcello,MFG13} Introduced the concept of {\em bounded depth proofs}
to approximate classical logic with nested applications of the structural
rule expressing the excluded middle axiom EM limited by a fixed natural number.}
intermediate proofs, i.e., $\IL$  proofs with 
{\em non-nested} applications of
the rules for  
classical disjunctive tautologies. 
This leads to $\lama$, an easy-to-use and yet expressive 
parallel extension of the simply typed
$\lambda$-calculus. 
Consisting only of {\em intuitionistic reductions} and (a simplified version of)
{\em basic cross reductions}, the reduction rules of $\lama$ always terminate, regardless of the reduction strategy. 

 Motivated by the tool Grace\footnote{Grace stands for GRAph-based Communication in Eden.}~\cite{Grace} which
allows programmers in the parallel functional language Eden to specify a 
 network of processes as a
 directed graph and provides special constructs to generate the actual
 network topology,
 we also provide an automatic
 procedure to extract $\lama$ typing rules from any communication topology in
 such a way that the typed terms can only communicate according to the
 topology. The idea of enforcing network topologies with types
 is also present in~\cite{Montesi} in the different context of
 the $\pi$-calculus~\cite{Milner} -- the most
 widespread formalism for {\em modeling} concurrent systems.

We consider classical disjunctive tautologies in the (intuitionistically equivalent) 
form

\centerline{$(\mathbb{A}_1 \impl \mathbb{A}_1 \et \mathbb{B}_1)\vel \dots \vel (\mathbb{A}_m \impl \mathbb{A}_m\et \mathbb{B}_m)\qquad \qquad $  \inlinelabel{eq:formulae}}

\noindent such that  all $\mathbb{B}_{i}$ are either $\bot$ or a conjunction of some
$\mathbb{A}_1, \dots, \mathbb{A}_{m}$.  
These axioms can indeed encode any reflexive directed graphs  
as follows: $\mathbb{A}_i$ represents a process in the network specified by the
graph, and $\mathbb{B}_i$ is the list of processes that are connected
to $\mathbb{A}_{i}$ by an outgoing arc, or $\bot$ if there are no such
processes. The intuitive reading of
$(\mathbb{A}_i \impl \mathbb{ A}_i\et \mathbb{ B}_i)$ is that all conjuncts in
$\mathbb{B}_i$, as well as $\mathbb{A}_i$ itself, 
can send messages to
$\mathbb{A}_i$. We allow $\mathbb{A}_i$ to send a message to itself in
case $\mathbb{A}_i$ wants to save it for later use, thus simulating a memory mechanism.
We establish a Curry--Howard correspondence between $\lama$ terms and
(fragments of) the intermediate logics 
obtained from the intuitionistic natural deduction calculus $\NJ$ by
allowing a non-nested use of the rules corresponding to 
axioms of the form \ref{eq:formulae}.  The
terms typed using these natural deduction rules contain as many parallel processes as
the number of premises.  For instance, the decorated version of the
rules for the axiom schema $ em= (\mathbb{A} \impl \mathbb{A}
\et \fal ) \vel(\mathbb{B} \impl \mathbb{B} \et \mathbb{A}) $ and
$ C_3= (\mathbb{A} \impl \mathbb{A} \et
\mathbb{B})\vel(\mathbb{B} \impl \mathbb{B} \et \mathbb{C})\vel
(\mathbb{C} \impl \mathbb{C} \et \mathbb{A})$, intuitionistically
equivalent to the excluded middle
law and $C_3$~\cite{L1982} are

\begin{footnotesize}
$$ \infer[(em)]{\pp{a}{u\p v}: C}{\infer*{u:C}{[a:A \impl A \et \fal ]
}\!\!&\!\!\infer*{v:C}{[a:B \impl B \et A]}}
\quad \quad \quad\infer[(C_3)]{\pp{a}{t \p u \p v }: D}{\infer*{t:D}{[a: A \impl A
\et B ] }\!\!&\!\!\infer*{u:D}{[a:B \impl B \et C]}\!\!&\!\!\infer*{v:D}{[a:C \impl C
\et A]}}$$
\end{footnotesize}with the (1-depth) restriction: $t, u$ and $v$ are 
a parallel composition of simply-typed $\lam$-terms that cannot communicate with each other.
The variable $a$ represents a private communication
channel that behaves similarly to the $\pi$-calculus 
operator $\nu$.  The typing rules establish how the communication channels connect the terms. For
example,  $(em)$ encodes the fact that the process $v$ can receive a
message of type $A$ from the process $u$, while the rule $(C_3)$
that $t$ can receive a message from $u$, $u$ from $v$ and $v$ from $t$. The behavior of these channels
during the actual communications is defined by the reduction rules of
the calculus.  If we omit reflexive edges, the communication topologies
corresponding to the reductions for the above rules are\\ 
\tikzstyle{proc}=[circle, minimum size=2mm, inner sep=0pt, draw]
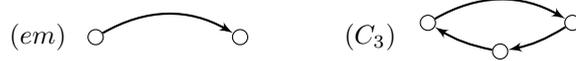
\begin{tikzpicture}[node distance=1.5cm,auto,>=latex', scale=0.19]

      \node [] () at (-52,0) {};
      
      \node [] () at (-37,0) {};

      \node [] () at (-32,0) {$ (em)$};

      \node [proc] (1) at (-28,0) {};

      \node [proc] (2) at (-18,0) {};

      \path[->] (1) edge [thick, bend left] (2);



      \node [] () at (-9,0) {$(C_3)$};

      \node [proc] (1) at (-5,1) {};

      \node [proc] (3) at (0,-1) {};

      \node [proc] (2) at (5,1) {};

      \path[->] (1) edge [thick, bend left] (2);

      \path[->] (2) edge [thick, bend left=10] (3);

      \path[->] (3) edge [thick, bend left=10] (1);
     \end{tikzpicture}
\\
In particular, $(em)$ implements the simplest message-passing mechanism, 
similar to that of
the asynchronous $\pi$-calculus~\cite{HT}, and $(C_3)$ 
cyclic communication among three processes.

The restriction (1-depth) enables us to define simple and yet
expressive reduction rules and to prove the strong normalization of
$\lama$. The resulting calculus is strictly more
 expressive than the simply typed $\lam$-calculus, and 
 can be used for interesting computational tasks (Sec.~\ref{sec:computing}). 

\section{The Typing System
of  $\lama$}
\label{sec:types}



%
%
%
%

$\lama$ extends the simply typed
$\lambda$-calculus with channels for multi-party communication and
reduction rules for message exchange.
Table~\ref{tab:type_nj} contains the type
assignments for $\lam$-terms, see e.g.~\cite{Girard} for details. Terms typed by such
rules are called \textbf{simply typed $\lambda$-terms}
and are denoted here by $t, u, v
\dots$
\begin{table}[t]
\centering
\begin{footnotesize}
  \begin{minipage}[l]{0.52\linewidth}
    \begin{minipage}[l]{0.5\linewidth}
      $\begin{array}{c} x^A: A \end{array} \qquad\quad \vcenter{\infer{
          \efq{P}{u}: P}{ u: \bot}} \;\; \text{ with $P$ atomic,
        $P \neq \bot$} $
    \end{minipage}
\bigskip

\begin{minipage}[l]{0.7\linewidth}
  $ \vcenter{\infer{
      \langle u,t\rangle: A \et B}{u:A & t:B}} \qquad\quad
  \vcenter{\infer{u\,\pi_0: A}{u: A\et B}} \qquad\quad \vcenter{\infer{u\,\pi_1: B}{u: A\et B}} $
\end{minipage}
\end{minipage}
  \begin{minipage}[l]{0.35\linewidth}
    $\vcenter{ \infer{\lambda x^{\scriptscriptstyle A} u: A\rightarrow
        B}{\infer*{u:B}{[x^{\scriptscriptstyle A}: A]}}} \qquad\quad \vcenter{\infer{tu:B}{ t:
        A\IMPL B & u:A}}$
  \end{minipage}
\end{footnotesize}
\medskip
\hrule
\smallskip
\caption{Type assignments for the simply typed $\lam$-calculus.}\label{tab:type_nj}

\vspace{-20pt}
\end{table}
Terms may contain
variables $x_{0}^{\scriptscriptstyle A},
x_{1}^{\scriptscriptstyle A}, x_{2}^{\scriptscriptstyle A}, \ldots$ of
type $A$ for every formula $A$.
Free and
bound variables of a proof term are defined as usual. We assume the
standard renaming rules and $\alpha$-equivalences that avoid the
capture of variables during reductions.

The typing rules of simply typed $\lam$-calculus, stripped of
$\lam$-terms, are the inference rules of Gentzen's natural
deduction system $\NJ$ for $\IL$. Actually, if $\Gamma= x_{1}: A_{1},
\shortdots , x_{n}: A_{n}$, and all free variables of a  term $t: A$ are in
$x_{1}, \shortdots , x_{n}$, from the logical point of view, $t$
represents an $\NJ$ derivation of $A$ from
the hypotheses $A_{1}, \shortdots, A_{n}$. We will thus write
$\Gamma\vdash t: A$.

\noindent \textbf{Notation.} $\rightarrow$  and $\et $ associate to
the right.
$\langle t_{1}, t_{2}, \ldots, t_{n}\rangle $ stands for 
$\langle t_{1}, \langle t_{2},$ $\ldots$ $\langle t_{n-1},
t_{n}\rangle$ $\ldots \rangle\rangle$, and $\langle t_{1}, t_{2},
\ldots, t_{n}\rangle \, \proj_{i}$, ($i=0, \ldots, n$) for 
$\langle t_{1}, t_{2}, \ldots, t_{n}\rangle\, \pi_{1}\ldots \pi_{1}
\pi_{0}$ containing the  projections that select the
$(i+1)$th element of the sequence.
$\non A$ is $A
\impl \bot $ and $\ver$ is $\bot \impl \bot$.

%
To type parallel terms that interact according to possibly
complex communication mechanisms, we build on ideas
from~\cite{ACGTCS,Marcello,DanosKrivine,CG2018} and base the Curry--Howard
correspondence for $\lama$ on a fragment of the axiomatic extensions of $\IL$ with
(intuitionistically equivalent versions of) classical disjunctive tautologies~\cite{DanosKrivine}. 
These axioms can be transformed into rules involving parallel
communicating sequents (i.e., hypersequents \cite{Avron91}), and
as shown in~\cite{ACGTCS} they lead to various communication schemata. 
The fragment considered in this paper consists of 1-depth proofs, 
whose notion is adapted from~\cite{Marcello,MFG13}. These are  
$\NJ$ proofs with {\em non-nested} applications of the natural deduction version of the corresponding
hypersequent rules, see~\cite{CG2018}. 
%
The use of this fragment drastically simplifies the reduction rules of $\lama$ w.r.t. the
calculi in~\cite{ACGTCS,ACGlics2017,ACG2018} and enables us to type channels with input/output directions, as in the $\pi$ calculus.
 
%
%
%

The class $\axx$ of axiom schemata that we consided here, that will be shown in Sec.~\ref{sec:topologies} 
to encode all communication topologies represented as reflexive directed graphs, is\vspace{.3pt}

 ${\color{white}A} \qquad\qquad\qquad\qquad\quad  (\mathbb{A}_1 \impl \mathbb{A}_1 \et \mathbb{B}_1)\vel \dots \vel (\mathbb{A}_m \impl \mathbb{A}_m\et
\mathbb{B}_m) $\vspace{.3pt}

\noindent where 
$\mathbb{A}_i \neq \mathbb{A}_j$ whenever $i\neq j$; and for any $i \in \{1,
\dots , m \}$, either $\mathbb{B}_i = \fal$ or $\mathbb{B}_i = \mathbb{A}_{k_1} \et \dots \et
\mathbb{A}_{k_p}$, 
 with 
$k_1<\ldots 
< k_p$.  
\begin{remark}
Each disjunct $\mathbb{A}_i \impl \mathbb{A}_i \et \mathbb{B}_i$ is
logically equivalent to $\mathbb{A}_i \impl
\mathbb{B}_i$. 
The redundant
occurrence of $\mathbb{A}_{i}$ is kept to type a memory mechanism for
input channels.\end{remark}As usual, an instance of an axiom schema
$\ax\in\axx$ is a formula obtained from $\ax$ by uniformly replacing
each propositional variable with an actual formula.
$\lama$ is obtained by the Curry--Howard correspondence applied
to $\NJ$ extended with non-nested applications
of the natural deduction rules for the axioms in $\axx$.
The type assignment for $\lama$ terms comprises the rules for the
simply typed $\lambda$-calculus in Table~\ref{tab:type_nj} and those
for the parallel operators in Table~\ref{tab:type_lama}.  
We denote the variables of simply typed $\lambda$-calculus as
$x^{\scriptscriptstyle A},$ $ y^{\scriptscriptstyle A}, $ $ z^{\scriptscriptstyle A} , \ldots,$ $ a^{\scriptscriptstyle A}, b^{\scriptscriptstyle A}, c^{\scriptscriptstyle A}, \dots$, $ \send{a}^{\scriptscriptstyle A}, \send{b}^{\scriptscriptstyle A}, \send{c}^{\scriptscriptstyle A}\dots$,  $ \get{a}^{\scriptscriptstyle A}, \get{b}^{\scriptscriptstyle A}, \get{c}^{\scriptscriptstyle A}\dots$ and,
whenever the type is not important, as $x, y, z, \ldots, a,
b, c \dots $  
We call \textbf{intuitionistic variables} the variables $x, y, z,
\dots$, which stand for terms; they are bound by the $\lambda$
operator.  The variables $a, \send{a}, \get{a}, \ldots$ are
called \textbf{channels} or \textbf{communication variables} and
represent  communication channels between parallel
processes: $a, b, c, \ldots$ are used as channel binders, $\send{a}, \send{b}, \send{c}, \ldots$ represent \textbf{output channels} that can transmit messages, while $\get{a}, \get{b}, \get{c}, \ldots$ \textbf{input channels} that can receive messages.
%
%
We denote $\ppp{a}{\ax}{u_{1} \p\ldots\p u_{m}}$ by
$\pp{a}{u_{1} \p\ldots\p u_{m}}$ when $\ax$ is clear from the context
or irrelevant. All free occurrences of $\send{a}$ and $\get{a}$ in $u_1 , \ldots , u_m$
are bound in $\pp{a}{u_1 \p \ldots \p u_m }$ and must have the types indicated by the inference rule ($\ax$).










\begin{table}[t]
\centering
\begin{footnotesize}

$\vcenter{\infer[(\contr)]{ t_1 \p \ldots \p t_n : A}{ t_1 : A
& \ldots  & t_n : A }}$
$\quad$ where $t_1 , \ldots , t_n $ are simply typed $\lam$-terms

\medskip

 $\vcenter{\infer[(\ax )]{\ppp{a}{\ax}{(u_{1} \p\ldots\p u_{n}) \p
\ldots \p (u_{p} \p\ldots\p u_{q})}: B} {\infer*{u_{1} \p\ldots\p
u_{n} :B}{[	\send{a}^{\scriptscriptstyle A_1 \impl A_1 \et B_1 }, \get{a}^{\scriptscriptstyle A_1 \impl A_1 \et B_1 }: A_1 \impl
A_1 \et B_1 ]} & \shortdots & \infer*{u_{p} \p\ldots\p
u_{q}:B}{[\send{a}^{\scriptscriptstyle A_m \impl A_m\et B_m}, \get{a}^{\scriptscriptstyle A_m \impl A_m\et B_m}: A_m \impl A_m\et B_m]} }}$
\medskip

where $ (A_1 \impl A_1 \et B_1)\vel \dots \vel (A_m \impl A_m\et
B_m)$ is an instance of $\ax\in\axx$ 

   \end{footnotesize}
\smallskip
\hrule
\caption{Type assignments for $\lama$.}\label{tab:type_lama}
\vspace{-20pt}
\end{table}The rule $(\contr)$ is useful for representing parallelism without communication. It is logically redundant, though, since it is an instance of $( \ax )$ with no channel occurrence.

From a computational perspective the rules
$(\ax)$ produce terms of the shape $\ppp{a}{\ax}{v_{1}\p\dots\p v_{m}}$ that
put in parallel $v_{1}, \dots, v_{m}$, which we call the \textbf{processes} of this term; 
each $v_{i}$ in turn has the shape $u_{1}\p\dots \p u_{k}$, where $u_{1}, \dots,
u_{k}$ are simply typed $\lambda$-terms called the \textbf{threads} of $v_i$. Processes can communicate through the channel $a$, whereas their threads  represent parallel independent subprograms that cannot interact with each other.
Informally, in order to establish
a communication channel connecting two terms $v_i$ and $v_j$, we require that
$\send{a}^{\scriptscriptstyle A_i \impl A_i\et B_i}$ occurs in $v_{i}$,  $\get{a}^{\scriptscriptstyle A_j \impl A_j\et
  B_j}$ occurs in $v_{j}$ and
   ${A}_{i}$ is in ${B}_j$.


On one hand, the argument $w$ of a channel application
$\send{a}^{\scriptscriptstyle A_{i} \impl A_i \et B_i}\, w$ will be
interpreted as a message of type $A_{i}$ that must be
\emph{transmitted}; on the other hand, the channel application
$\get{a}^{\scriptscriptstyle A_{j} \impl A_j\et B_j}\, t$ will
\emph{receive} a batch of messages of type $B_{j}$ containing $w$
that will replace the whole channel application $\get{a}^{\scriptscriptstyle
A_{j} \impl A_j \et B_j}\, t$ upon reception. This is the reason why the direction of the communication and the direction of $\impl$ are reversed. In general, each
channel application $\send{a} v$ first transmits $v$ and, immediately after that, starts to listen on the same channel by reducing to $\get{a} v$  and waiting for a message that will replace the whole term $\get{a} v$. To formalize the relation between a process $v_i$ and all the
processes $v_j$ such that $v_i$ can send messages to $v_j$, we need to
look at the structure of the axiom schema $\ax$, because its instances
may lose information about its general shape. For this purpose, we
introduce the concept of outlink.
\begin{definition}[Outlinks] \label{def:outlink}
Let  $\ppp{a}{\ax}{v_{1}\p\dots\p v_{m}}$ be a term, where 
$\ax = (\mathbb{A}_1 \impl \mathbb{A}_1 \et \mathbb{B}_1)\vel
\shortdots \vel (\mathbb{A}_m \impl \mathbb{A}_m\et \mathbb{B}_m)$.
For any $i, j \in \{1, \shortdots , m \}$ and $i\neq j$,
we say that the term $v_i $ is \textbf{outlinked} to the term $v_j$
if $\, \mathbb{B}_j = \mathbb{A}_{k_1} \et \shortdots \et
\mathbb{A}_i \et \shortdots \et \mathbb{A}_{k_p}$. 
\end{definition}
\begin{remark}
The restriction of (1-depth) proofs forces the derivations of the premises of the rule $(\ax)$ to be $\IL$ proofs; by removing this restriction, the rule is equivalent to an (unrestricted) instance of the axiom $\ax$, see~\cite{CG2018}. 
\end{remark}

\section{Communications in $\lama$}\label{sec:red}




In Sec.~\ref{sec:types} we showed how to install communication
channels connecting processes, we present now the reduction rules of
$\lama$ that implement the actual communications. 
Since these communications are higher-order, they 
can transmit arbitrary simply typed
$\lambda$-terms as messages, provided their free variables are not
bound in the surrounding context. Unlike in \cite{ACGTCS}, messages are thus not restricted to values and communication channels have directions.  The reduction rules of $\lama$
comprise two groups of rules: those for the simply typed
$\lambda$-calculus (\emph{intuitionistic reductions}), and those
that deal with process communication (\emph{cross
reductions}). \\ 
\textbf{Intuitionistic Reductions.} The usual computational rules
for the simply typed $\lam$-calculus represent the operations of
applying a function to an argument and extracting a component from a pair~\cite{Girard}. From the logical point of view, they are the standard
Prawitz' reductions~\cite{Prawitz} of the natural deduction calculus $\NJ$ for $\IL$.\\
\noindent \textbf{Cross Reductions}.  Their goal is to implement
communication, namely to transmit programs in the form of simply typed
$\lambda$-terms. A very simple example of cross reduction for the $(em)$ topology shown in Section \ref{sec:intro} is the following:\vspace{.3pt}

$\qquad\qquad \qquad\qquad\quad\;\; \pp{a}{\mathcal{C} [\send{a}v] \p \mathcal{D}[\get{a}w]}\;\mapsto \;\pp{a}{\mathcal{C}[\get{a}v] \p \mathcal{D} [\langle w, v\rangle ]}$
 
\noindent where $\mathcal{C} [\; ], \mathcal{D}[\;],v$ and $w$ contain no channels. In general,  since $\lama$ terms may contain more than one output channel
application, we first have to choose which application will transmit
the next message.
For example, here we have two communicating processes, each containing
two threads:\vspace{.3pt}

$(\ast) \qquad\qquad\qquad\quad \pp{a}{ \;\;  ( \get{a}\, r
\p \send{a} ( x( \send{a} \, s)) )\;\;  \p\;\;  ( \get{a}\, u \p \get{a} ( y (\get{a} w) ) ) \;\; }$\vspace{.3pt}

\noindent where $r, s , u, w$ are simply typed $\lam$-terms not
containing $a$. The first process $\get{a}\, r \p \send{a} ( x( \send{a} \, s))$ contains
two occurrences of the output channel $\send{a}$. Let us focus on its second
thread $  \send{a} ( x( \send{a} \, s)) $.  The channel application $\send{a} ( x( \send{a} \, s))$
cannot transmit the message $x (\send{a}\, s)$, because the channel $a$ might
be used with a different type in the second process, so type
preservation after reduction would fail. Hence the only
possibility here is to choose the second channel application $\send{a}\, s$
as the one containing the output message, in this case $s$. In
general, to make sure that the message does not contain the channel
$a$, it is enough to choose as application occurrence that contains
the output message the \emph{rightmost occurrence} of the  channel ${a}$
in the whole process, provided it occurs as an actual output channel $\send{a}$. Hence, in any process, the rightmost thread that
contains the channel $\send{a}$ may contain the message. If instead the \emph{rightmost occurrence} of $a$ is of the form $\get{a}$, the process has no message to send.

The second choice is  which occurrence of an input channel application should receive the current message. For example, in the term $(\ast)$ above,
the second process $\get{a}\, u \p \get{a} ( y (\get{a} w) )$ contains three occurrences
of $\get{a}$. Since a channel can both send and receive, and has usually sent something  before receiving, we are led to choose again the rightmost
occurrence of a channel application as receiver, provided it is an input channel.
Nonetheless, since the threads of a single process do not communicate
with each other, it is best to let all of them receive the message in
correspondence of their \emph{locally rightmost}  channel
application. Thus in the example, both $\get{a}\, u $ and $\get{a}\, w$ will
receive messages.
 

 The third choice is what to do with the arguments of a receiving
channel. After a few programming examples, as those in
Section~\ref{sec:computing}, it is natural to
convince oneself that it is better to keep the arguments, a feature
that we call \emph{memory}. In the previous term, when $\get{a}\, u$ and $\get{a}\,
w$ receive $s$, they will be replaced respectively, with $\langle u,
s\rangle$ and $\langle w, s\rangle$.
 

Continuing our example and summing up, we will have the following
reduction of $(\ast)$:\vspace{.3pt}

\centerline{{\small $ \pp{a}{ \; ( \get{a}\, r
\p \send{a} ( x( \send{a} \, s)) )\;\;  \p\;\;  ( \get{a}\, u \p \get{a} ( y (\get{a} w) ) ) \; } \quad \mapsto \quad \pp{a}{ \; ( \get{a}\, r \p \send{a} (
x( \get{a} \, s)))\; \p\; ( \lan u , s \ran \p \get{a
} ( y \lan w , s \ran ) ) \;
}$}}

\noindent As we can see, the message $s$ in
correspondence of the rightmost occurrence of $\send{a}$ in $ \send{a} ( x( \send{a} \,
s))$ is transmitted by the first process to the rightmost application
of $\get{a}$ in each one of the two threads $\get{a}\, u$ and $\get{a} ( y (\get{a} w) )$ of
the second process; at the same time, the output channel application $\send{a}\, s$ has turned  into the input channel application $\get{a}\, s$.

The fourth choice to make is which processes should receive messages
and which processes should send them.  As anticipated in the previous
section, we are guided by the typing. Let us consider for instance the
term $ \pp{a}{ \; x \, (\get{a}^{\scriptscriptstyle A \impl A \et B} \,s) \p
\; y\, (\send{a}^{\scriptscriptstyle B \impl B \et A}\, t) \; }$, where $s$
and $ t$ are specific simply typed $\lam$-terms not containing $a$, while $x:A \et B \impl C$
and $y:B \et A \impl C$.  A cross reduction rule corresponding to this
typing rule admits communication in two directions, from left to right
and from right to left, because, according to the types, the second
process can receive messages from the first and viceversa. The actual direction of the message is determined however by the rightmost occurrences of the channel $a$ in the two processes. In this case, the rightmost occurrence of $a$ in the first process is an input channel, while the rightmost occurrence of $a$ in second process is an output channel.   Hence the
reduction is from right to left and is $ \quad \pp{a}{ \; x \,
(\get{a}^{\scriptscriptstyle A \impl A \et B} \,s) \p \; y\,
(\send{a}^{\scriptscriptstyle B \impl B \et A}\, t) \; } \;\mapsto\; \pp{a}{
\; x \, \lan s, t\ran\p y\, (\get{a}^{\scriptscriptstyle B \impl B \et A}\,
t) \ \; }\;$


To define communication reductions, we need to introduce two
kinds of contexts: one for terms that can communicate, one for terms
which are in parallel but cannot communicate.

\begin{definition}[Simple Parallel Term]\label{def:simpparterm} A
\textbf{simple parallel term} is a $\lama$ term $t_{1}\p\ldots \p t_{n}$, where each
$t_{i}$, for $1\leq i\leq n$, is a simply typed $\lambda$-term.\end{definition}


\begin{definition}\label{defisimplec}\label{def:simpleparallelcont} A
\textbf{context} $\mathcal{C}[\ ]$ is a $\lama$ term with some fixed
variable $[\; ]$ occurring 
\begin{itemize}
\item A \textbf{simple context} is a context which is a simply typed
$\lambda$-term.
\item A \textbf{simple parallel context} is a context which is a
simple parallel term.
\end{itemize} For any $\lama$ term $u$ of the same type of $[\; ]$,
$\mathcal{C}[u]$ denotes the term obtained replacing $[\; ]$ with $u$
in $\mathcal{C}[\ ]$, \emph{without renaming bound variables}.\end{definition}
We explain the general case of the cross reduction.  The rule
identifies a single process as the receiver and possibly several
processes as senders. Once the receiving process is fixed, the senders
are determined by the axiom schema $\ax$, that is instantiated by the type of the
communication channel occurring in the receiving process.
In particular, in the term\vspace{.3pt}

$(\star) \qquad\qquad  \ppp{a}{\ax} {\shortdots \!\! \p \mathcal{C}_1[\send{a} w_1 ]\p
 \! \shortdots \! \!  \p (\shortdots\p{\mathcal D}_1[\get{a} \, v_1]\p \!
\shortdots \! \! \p {\mathcal D}_n [\get{a} \, v_n] \p\shortdots) \p \!
\shortdots \! \!  \p \mathcal{C}_p[\send{a} w_p ]\p \! \shortdots
} $\vspace{.3pt}

\noindent the processes $\mathcal{C}_1[\send{a} w_1\ ], \dots,
\mathcal{C}_p[\send{a} w_p\ ]$ are the senders and $(\shortdots\p{\mathcal D}_1[\get{a} \, v_1]\p \! \shortdots \! \! \p
{\mathcal D}_n [\get{a} \, v_n]\p\shortdots) $ in its entirety is the receiver. Formally,
$\mathcal{C}_1[\send{a} w_1\ ], \dots, \mathcal{C}_p[\send{a} w_p\ ]$ are all the
process \emph{outlinked} to the process
$(\shortdots\p{\mathcal D}_1[\get{a} \, v_1]\p \! \shortdots \! \! \p
{\mathcal D}_n [\get{a} \, v_n]\p\shortdots) $, see Definition \ref{def:outlink}. In this latter process, we
have displayed  the threads that actually contain the channel $a$:
all of them will receive the messages.  Consistently with our choices,
the displayed occurrences of $a$ are rightmost in each ${\mathcal
D}_j[\get{a} \, v_j]$ and in each $\mathcal{C}_j[\send{a}\, w_j]$. The processes
containing $w_1, \ldots , w_p$ send them to all the rightmost
occurrences of $\get{a}\, v_{1}, \dots, \get{a}\, v_{n}$ in the processes ${\mathcal D}_1 , \ldots ,
{\mathcal D}_n$:\vspace{.3pt}

\centerline{ {\small$ (\star) \quad
\mapsto\quad \ppp{a}{\ax} { \shortdots \!\! \p \!\mathcal{C}_1[a w_1
]\p \! \shortdots\!\! \p\!\! (\shortdots \!\! \p \! {\mathcal
D}_1[\lan v_1, w_1, \shortdots , w_p \ran] \p \! \shortdots\!\! \p \!
{\mathcal D}_n [\lan v_n, w_1, \shortdots , w_p \ran]\p \! \shortdots
)\!\! \p\!  \shortdots\!\! \p \! \mathcal{C}_p[a w_p ] \p \!
\shortdots \! }$}}\vspace{.3pt}

\noindent provided that, for each $w_j$, its
free variables are free in $\mathcal{C}_j[a\, w_j]$: this condition is
needed to avoid that bound variables become free, violating the
Subject Reduction. Whenever $w_{j}$ is a closed term -- executable
code -- the condition is automatically satisfied.   As we can see, the
reduction retains all terms $v_1 , \ldots , v_n $ occurring in
${\mathcal D}_1 , \ldots , {\mathcal D}_n$ before the
communication.\begin{remark}Unlike in
$\pi$-calculus, channel \emph{occurrences} that send messages are not immediately consumed in $\lama$,
because cross reductions adopt the perspective of the receiver rather
than that of the senders.
 However, after transmission, the output
channel occurrence will be turned immediately into an input channel occurrence, which will be consumed when the
process is selected as current receiver. 
\end{remark}
The last choice is what to do with the threads or processes that do
not contain any communication channel. The idea is that whenever a
term contains no channel occurrence, it has already reached a result,
as it does not need to interact with the context.  Hence at the end of
the computation we select some of the processes that have reached
their own results and consider them all together the global result of
the computation.  Thus we introduce the simplification reduction in
Table~\ref{tab:red}, which also displays all the reduction rules of
$\lama$.
%
%
%
%
%
%
%
%
%
%
%
%
%
%
\begin{table*}[t]
\begin{footnotesize}\begin{flushleft} \textbf{Intuitionistic
Reductions} $\quad\quad (\lambda x^{\scriptscriptstyle A}\, u)t\mapsto
u[t/x^{\scriptscriptstyle A}] \qquad \pair{u_0}{u_1}\,\pi_{i}\mapsto
u_i \;  \mbox{ for $i=0,1$} $\vspace{-5pt}\end{flushleft}\textbf{Cross Reductions: Communication}\vspace{-5pt}
  \begin{center} $\ppp{a}{\ax} {\shortdots\!\! \p \mathcal{C}_1[\send{a} w_1
]\p\!\shortdots\!\! \p \! \!(\shortdots \!\!\p {\mathcal D}_1[\get{a} \,
v_1]\p \!\shortdots\!\! \p {\mathcal D}_n [\get{a} \, v_n]\p \!  \shortdots
) \! \! \p \! \shortdots\!\! \p \mathcal{C}_p[\send{a} w_p ]\p \! \shortdots
} $$\quad \mapsto$\\$\ppp{a}{\ax} { \shortdots \!\! \p \!\mathcal{C}_1[\get{a}
w_1 ]\p \! \shortdots\!\! \p\!\! (\shortdots \!\! \p \!  {\mathcal
D}_1[\lan v_1, w_1, \shortdots , w_p \ran] \p \!  \shortdots\!\! \p \!
{\mathcal D}_n [\lan v_n, w_1, \shortdots , w_p \ran]\p \! \shortdots
)\!\! \p\!  \shortdots\!\! \p \!  \mathcal{C}_p[\get{a} w_p ] \p \!
\shortdots \! }$\end{center} where 
$\mathcal{C}_1[\send{a}\, w_1], \dots, \mathcal{C}_p[\send{a}\, w_p]$ are all the
processes outlinked to the process $ ( \shortdots \p{\mathcal D}_1[\get{a} \, v_1]
  \p  \shortdots  \p {\mathcal D}_n  [\get{a} \, v_n] \p \shortdots ) $; the displayed occurrences of $a$  are rightmost in each
${\mathcal D}_j[\get{a} \, v_j]$ and in each
$\mathcal{C}_j[\send{a}\, w_j]$; 
each
$\mathcal{D}_j$ is a simple context and each $\mathcal{C}_j$ is a
simple parallel context; the free variables of each $w_j$ are
free in $\mathcal{C}_j[\send{a}\, w_j]$; 
\\
\vspace{-5pt}

\textbf{Cross Reductions: Simplification}$\quad$\begin{minipage}[h]{0.6\linewidth}
  $\quad \ppp{a}{\ax}{ (u_1 \p \! \shortdots \!\! \p u_n) \p \!
    \shortdots \!\! \p (u_{m} \p \! \shortdots \!\! \p u_{p}) }
  \mapsto u_{i_1} \p \! \shortdots \!\!  \p u_{i_q}$

  whenever $ u_{i_1} , \shortdots , u_{i_q}$ do not contain $a$ and
  $1\leq i_1 < \! \shortdots \!\! < i_q \leq p$.
\end{minipage}
\end{footnotesize}
\hrule
\caption{Reduction Rules for  $\lama$.%
}
\label{tab:red}
\end{table*}As usual, we adopt the reduction schema:
$\mathcal{C}[t]\mapsto \mathcal{C}[u]$ whenever $t\mapsto u$ and for
any context $\mathcal{C}$. We denote by $\mapsto^{*}$ the
reflexive and transitive closure of the one-step reduction
$\mapsto$.


\begin{theorem}[Subject Reduction]\label{subjectred}
If $t : A$ and $t \mapsto u$, then $u : A$ and all the free variables
of $u$ appear among those of $t$.
\end{theorem} 
\begin{proof} 
The only not trivial case is that of 
%
cross reductions. Assume the step is as follows:\begin{center}{\footnotesize
    $\ppp{a}{\ax} {\shortdots\!\! \p \mathcal{C}_1[\send{a} w_1
      ]\p\!\shortdots\!\! \p \! \!(\shortdots \!\!\p {\mathcal
        D}_1[\get{a}^{\scriptscriptstyle A _i \impl A_i \et B_i } \, v_1]\p \!\shortdots\!\! \p
      {\mathcal D}_n [\get{a}^{\scriptscriptstyle A _i \impl A_i \et B_i } \, v_n]\p \!
      \shortdots ) \! \! \p \! \shortdots\!\! \p \mathcal{C}_p[\send{a} w_p
      ]\p \! \shortdots } $
\\    $\mapsto$\\
    $\ppp{a}{\ax} { \shortdots \!\! \p \!\mathcal{C}_1[\get{a} w_1 ]\p \!
      \shortdots\!\!  \p\!\! (\shortdots \!\! \p \! {\mathcal
        D}_1[\lan v_1, w_1, \shortdots , w_p \ran] \p \!
      \shortdots\!\! \p \! {\mathcal D}_n [\lan v_n, w_1, \shortdots ,
      w_p \ran]\p \! \shortdots )\!\!  \p\!  \shortdots\!\! \p \!
      \mathcal{C}_p[\get{a} w_p ] \p \!  \shortdots \! }$}
  \end{center} As the type of the channel $a$ is an instance
$(A_1 \impl A_1 \et B_1)\vel \! \shortdots \!\! \vel (A_m \impl A_m\et
B_m) $ of the schema $\ax = (\mathbb{A}_1 \impl \mathbb{A}_1 \et
\mathbb{B}_1)\vel \shortdots \vel (\mathbb{A}_m \impl \mathbb{A}_m\et
\mathbb{B}_m) $ where $\mathbb{B}_i = \mathbb{A}_{k_1} \et \shortdots
\et \mathbb{A}_{k_p}$, and as $ w_1 : A_{k_1} , \shortdots , w_p :
A_{k_p} $ where $A_{k_1} , \shortdots , A_{k_p}$ instantiate
$\mathbb{A}_{k_1} , \shortdots , \mathbb{A}_{k_p}$, then
$(\shortdots \!\! \p \! {\mathcal D}_1[\lan v_1, w_1,
\shortdots , w_p \ran] \p \! \shortdots\!\! \p \! {\mathcal D}_n [\lan
v_n, w_1, \shortdots , w_p \ran]\p \! \shortdots )$ is well defined
and its type is the same as that of $(\shortdots \!\!\p {\mathcal
D}_1[\get{a}^{\scriptscriptstyle A _i \impl A_i \et B_i } \, v_1]\p \!\shortdots\!\! \p
{\mathcal D}_n [\get{a}^{\scriptscriptstyle A _i \impl A_i \et B_i } \, v_n]\p \!  \shortdots
)$. Hence the term type term does not change.

Since the displayed occurrences of $a$ are rightmost in each
$\mathcal{C}_j[\send{a}\, w_j]$ and thus $a$ does not occur in $w_j$, no
occurrence of $a$ with type different from $A _i \impl A_i \et B_i$ occurs in the terms ${\mathcal D}_1[\lan
v_1, w_1, \shortdots , w_p \ran] , \shortdots , {\mathcal D}_n [\lan
v_n, w_1, \shortdots , w_p \ran]$.
Finally, the free variables of each $w_j$ are free also in
$\mathcal{C}_j[\send{a}\, w_j]$, and thus no new free variable is created by
the reduction.

\end{proof} 
\begin{remark}
The communication reductions of $\lama$ allow
processes to communicate independently of their order, and hence the parallel composition
needs not to be commutative.
\end{remark}

\section{From Communication Topologies to Programs} 
\label{sec:topologies}


We present a method for automatically extracting $\lama$ typing rules
from graph-specified communication topologies.  Given a
directed, reflexive graph $G$ whose nodes and edges 
represent
respectively processes and communication channels, we describe how to
transform it into an axiom schema $\ax \in\axx$ corresponding to a typing
$(\ax)$ rule for $\lama$ terms. We will show that this typing rule encodes a process topology
which exactly mirrors the graph $G$: two processes may communicate if
and only the corresponding graph nodes are connected by an edge and
the direction of communication follows the edge. In other
words, the edges of the graph correspond to the \emph{outlinked}
relation between processes of Definition~\ref{def:outlink}.



\begin{procedure}\label{proc:topology}
Given a directed reflexive graph $G = (V, E)$ with $k=|V|$, the axiom
schema $\ax$ encoding $G$ is $\mathbb{C}_{1}\lor\dots\lor
\mathbb{C}_{k}$ such that for each $n\in\{1, \dots, k\}$:
 \begin{itemize}
\item $\mathbb{C}_{n}=\mathbb{A}_n \impl \mathbb{A}_n \wedge
\mathbb{A}_{i_1} \et \ldots \et \mathbb{A}_{i_m}$, if the $n$-th node
in $G$ has incoming edges from the non-empty list of pairwise distinct
nodes $i_1, \ldots , i_m\neq n $;
      \item $\mathbb{C}_{n}= \mathbb{A}_n \impl \mathbb{A}_n \et
\fal$, if the $n$-th node in $G$ has only one incoming edge.
      \end{itemize}\end{procedure}
We show an example and then state the correspondence between
graphs and $\lama$ reductions.

%

\begin{example}\label{ex:topology}
Consider the axiom below corresponding to the graph\begin{center}
  \tikzstyle{proc}=[circle, minimum size=4.5mm, inner sep=0pt, draw]
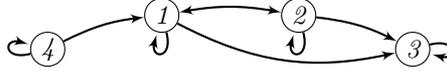
  \begin{tikzpicture}[node distance=1.5cm,auto,>=latex', scale=0.3]
    \node [proc] (4) at (-7,1.5) {4};
    
    \node [proc] (3) at (9,1.5) {3};

    \node [proc] (2) at (4,3) {2};

    \node [proc] (1) at (-2,3) {1};

    \path[<->] (1) edge [thick, bend left=10] (2);

    \path[->] (1) edge [thick, bend right=20] (3);

    \path[->] (2) edge [thick, bend left=10] (3);

    \path[->] (4) edge [thick, bend left=10] (1);

    \path[->] (1) edge [thick, loop below] (1);
    
    \path[->] (2) edge [thick, loop below] (2);

    \path[->] (3) edge [thick, loop right] (3);

    \path[->] (4) edge [thick, loop left] (4);
  \end{tikzpicture}\end{center}
  $$(\mathbb{A}_1\impl \mathbb{A}_1 \et \mathbb{A}_2 \et \mathbb{A}_4)
\vel (\mathbb{A}_2\impl \mathbb{A}_2 \et \mathbb{A}_1) \vel
(\mathbb{A}_3\impl \mathbb{A}_3 \et \mathbb{A}_1 \et \mathbb{A}_2)\vel
(\mathbb{A}_4\impl \mathbb{A}_4 \et
\fal) $$\end{example}\vspace{-5pt}The rule extracted from a graph
forces communications to happen as indicated by the edges of the
graph. Indeed the following holds.
\begin{proposition}[Topology Correspondence] For any directed reflexive
graph $G$ and term $t:= \pp{a}{u_1 \p \ldots \p u_m}$ typed using the rule
corresponding to the axiom extracted from $G$ by
Procedure~\ref{proc:topology}, there is a cross
reduction for $t$ that transmits a term $w$
from $u_x$ to $u_y$ if and only if $G$ contains an edge from $x$ to
$y$.\end{proposition}

\section{The Strong Normalization Theorem}\label{sec:norm}

We prove the strong normalization theorem for $\lama$: any reduction of every  $\lama$ term ends in a finite number of steps into a normal form. This means that the computation of any typed $\lama$ term 
 always terminates independently of the chosen reduction strategy. 
Instead of struggling directly with the complicated combinatorial properties of process communication, as done for the weak normalization in~\cite{ACG2018,ACGlics2017}, we reduce the strong normalization of $\lama$ to the strong normalization of a non-deterministic reduction relation over simply typed $\lambda$-terms. The idea is to  simulate communication by non-determinism, a technique inspired by~\cite{AschieriZH}. 

\begin{definition}[Normal Forms and Strongly Normalizable Terms]\label{def:formnorm}\mbox{}
 \begin{itemize}
\item  A \textbf{redex} is a term $u$ such that $u\mapsto v$ for some
  $v$ and reduction in Table~\ref{tab:red}. 
A term $t$ is called a \textbf{normal form} or, simply, \textbf{normal}, 
if $t$ is not a redex.
\item  
A finite or infinite sequence of terms
$u_1,u_2,\ldots,u_n,\ldots$ is said to be a \textbf{reduction} of $t$, if
$t=u_1$, and for all  $i$, $u_i \mapsto
u_{i+1}$.
 A  term $u$ of $\lama$ is  \textbf{normalizable} if there is a finite reduction of $u$ whose last term is normal and is \textbf{strong normalizable} if every reduction of $u$ is finite. With $\sn$ we denote the set of the strongly normalizable terms of $\lama$.\end{itemize}\end{definition}



\subsection{Non-deterministic Reductions}

Strong normalization for parallel $\lambda$-calculi is an intricate
problem.  Decreasing complexity measures are quite hard to find  and indeed those introduced in \cite{ACG2018,ACGlics2017} for
proving normalization of fragments of $\lama$ fail in case of strong
normalization.  Given a term $\pp{a}{u_{1}\p\dots\p u_{n}}$, one would
like to measure its complexity as a function of the complexities of
the terms $u_{i}$, for example taking into account the number of
channel occurrences and the length of the longest reduction of
$\lambda$-redexes in $u_{i}$. However, when $u_{i}$ receives a message
its code may drastically change and both those numbers may
increase. Moreover, there is a potential circularity to address:
channels send messages and may generate new $\lambda$-calculus
redexes in the receivers; $\lambda$-calculus redexes may duplicate
channel occurrences, generating even more communications.

To break this circularity, we use a radically different
complexity measure. The complexity of a process $u_{i}$ should take
into account all the possible messages that $u_{i}$ may receive. If
$u_{i}=\mathcal{D}[\get{a}\, t ]$, after receiving a message it becomes
$\mathcal{D}[\lan t , s \ran ]$, where $s$ is an \emph{arbitrary}
simply typed $\lambda$-term from the point of view of $u_{i}$. We thus
create a reduction relation over simply typed $\lambda$-terms that
simulates the reception ``out of the blue'' of this kind of
messages. Namely, we extend the reduction relation of
$\lambda$-calculus by two new rules: i) $\send{a}^{T}\redn \get{a}^{T}$; ii) $\get{a}^{T} \!\redn\! t$ for every channel
$a^{T}$ and simply typed $\lambda$-term $t: T$ that does not contain channels.  In order to simulate  the reception of an arbitrary
batch $w$ of messages,  $t$ will be
instantiated as $\lam x \, \lan x , w\ran $ in the proof of
Th.~\ref{SNT}. With these reductions,
$\lambda$-terms that do not contain channels are the usual
deterministic ones.

\begin{definition}[Deterministic simply typed $\lambda$-terms] A
simply typed $\lambda$-term $t$ is called \textbf{deterministic}, if $t$
does not contain any channel occurrence.
\end{definition}

\begin{definition}[The non-deterministic reduction
relation $\redn$] The reduction relation $\redn$ over simply typed
$\lambda$-terms is defined as extension of the relation $\mapsto$ as
follows:\begin{footnotesize}$$ (\lambda x^{\scriptscriptstyle A}\, u)\,
t\redn u[t/x^{\scriptscriptstyle A}] \quad \quad
\pair{u_0}{u_1}\,\pi_{i}\redn u_i, \mbox{ for $i=0,1$} \quad \quad 
\send{a}^{T} \redn \get{a}^{T}$$
$$\get{a}^{T} \redn t, \ \text{ for every channel $a^{T}$ and deterministic simply typed $\lambda$-term $t: T$}$$\end{footnotesize}and as usual we close by contexts: $\mathcal{C}[t]\redn \mathcal{C}[u]$ whenever
$t\redn u$ and $\mathcal{C}[\,]$ is a simple context. We denote by $\redn^{*}$
the reflexive and transitive closure of the one-step
reduction $\redn$.\end{definition}
The plan of our proof will be to prove the strong normalization of simply typed $\lambda$-calculus with respect to the reduction $\redn$ (Corollary \ref{cor:snn}) and then derive the strong normalization of $\lama$ using $\redn$ as the source of the complexity measure (Theorem \ref{Adequacy Theorem}). The first result 
will be proved by the standard Tait--Girard reducibility technique (Definition \ref{def:red}).

We define $\snn$ to be the set of strongly normalizing simply typed
$\lambda$-terms with respect to the non-deterministic reduction
$\redn$. The reduction tree of a strongly normalizable term with
respect to $\redn$ is no more finite, but still well-founded. It is
well-known that it is possible to assign to each node of a
well-founded tree an ordinal number, in such a way that it decreases
passing from a node to any of its children. We will call the \emph{ordinal
size} of a term $t\in\snn$ the ordinal number assigned to the root of
its reduction tree and we denote it by $h(t)$; thus, if $t\redn u$,
then $h(t)>h(u)$. To fix ideas, one may define $h(t):=\mathsf{sup}\{
h(u)+1\ |\ t\rightsquigarrow u\}$.

\subsection{Reducibility and Properties of Reducible Terms}\label{section-reducibility}

We define a notion of reducibility for simply typed $\lambda$-terms  with respect to the reduction $\redn$. As usual, we prove that every reducible term is strongly normalizable w.r.t. $\redn$ and afterwards that all simply typed $\lambda$-terms are reducible. The difference with the usual reducibility proof is that we first prove that deterministic simply typed $\lambda$-terms are reducible (Th.~\ref{Adequacy Theorem Mini}), which makes it possible to prove  that channels are reducible (Prop.~\ref{prop:channels}) and finally that all terms are reducible (Th.~\ref{Adequacy Theorem}). This amounts to prove the usual Adequacy Theorem twice. 
This stratification of the proof reminds of the reducibility technique employed in~\cite{Amadio} for proving weak normalization of a concurrent $\lambda$-calculus with shared memory.

\begin{definition}[Reducibility]\label{def:red}
\label{definition-reducibility}
Assume $t: C$ is a simply typed $\lambda$-term. We define the relation $t\red C$ (``$t$ is reducible of type $C$'') by induction and by cases according to the form of $C$:\\
\textbf{1.}
$t\red \emp{}$, with $\emp{}$ atomic, if and only if 
$t\in\snn$\\
\textbf{2.} $t\red {A\wedge B}$ if and only if $t\, \pi_0 \red {A}$ and $t\, \pi_1\red {B}$\\
\textbf{3.} $t\red {A\rightarrow B}$ if and only if for all $u$, if $u\red
{A}$, then $tu\red {B}$\end{definition}
We show that the set of reducible terms for a formula $C$ is a  reducibility candidate~\cite{Girard}. Neutral terms
are defined as terms that are not ``values'' and need to be further reduced.


\begin{definition}
A simply typed $\lambda$-term  not of the form $\lambda x\, u$ or $\pair{u}{t}$
is neutral.
\end{definition}

\begin{definition}[Reducibility Candidates] Extending the approach of \cite{Girard}, we define three properties
  of reducible terms $t$:
\cruno\ If $t\red A$, then $t\in \snn$, 
\crdue\ If $t \red A$ and $t\redn^{*} t'$, then $t' \red A$, and
\crtre\  If $t$ is neutral and  for every $t'$, $t\redn t'$ implies $t'\red A$, then $t \red A$.

\end{definition}
We show that every term $t$ possesses the reducibility candidate properties. The arguments for 
\cruno, \crdue, \crtre{} are  standard (see~\cite{Girard}).

\begin{proposition}[\cite{Girard}]
Every simply typed $\lambda$-term
satisfies \cruno, \crdue, \crtre.\end{proposition}
The next task is to prove
that all introduction rules of simply typed $\lambda$-calculus define
a reducible term from a list of reducible terms for all premises.
\begin{proposition}\label{proposition-somecases}\mbox{}
(1) If for every $t\red A$, $u[t/x]\red B$, then  $\lambda x\, u\red A\rightarrow B$ and
(2) If $u\red A$ and $v\red B$, then $\pair{u}{v}\red  A\land B$.

\end{proposition}
%
%
%
 



\begin{proof}
As in~\cite{Girard}, using $\cruno$, $\crdue$ and $\crtre$.
\end{proof}

\subsection{The Mini Adequacy Theorem}

We prove that simply typed $\lambda$-terms that do not contain channels are reducible.

\begin{theorem}[Mini Adequacy Theorem]\label{Adequacy Theorem Mini}\label{thm:MiniAdeq}
Suppose that $w: A$ is a deterministic simply typed $\lambda$-term, with intuitionistic free variables among $x_1: {A_1},\ldots,x_n:{A_n}$. For all terms $t_1, \ldots, t_n$ such that
$\text{ for  $i=1,\ldots, n$, }t_i\red A_i$ 
we have
$w[t_1/x_1\cdots t_n/x_n]\red A$
\end{theorem}
\begin{proof}
As the proof of Th.~\ref{Adequacy Theorem} without case 2., which concerns channels. 
\end{proof}







\begin{corollary}[Mini Strong Normalization of $\redn$]\label{cor:sndt} If $t: A$ is a deterministic simply typed $\lambda$ term, then $t\red A$ and $t\in\snn$.
\end{corollary}
\begin{proof}
Assume $x_1: {A_1},\ldots,x_n:{A_n}$ are all the intuitionistic free variables of $t$ and thus all its free variables. By \crtre, one has $x_{i}\red A_{i}$, for $i=1,\ldots, n$. 
From Theorem \ref{Adequacy Theorem Mini}, we derive $t\red A$. From \cruno,  we conclude that $t\in\snn$.\end{proof}By the Mini Adequacy Theorem we can prove that channels are reducible.

\begin{proposition} [Reducibility of  Channels]\mbox{}
\label{prop:channels}
 i) For every input channel $\get{a}: T$, $\get{a} \red T$, and 
 ii) For every output channel $\send{a}: T$, $\send{a} \red T$
\end{proposition}
\begin{proof}
i) Since $\get{a}$ is neutral, by \crtre{} it is enough to show that for all $u$ such that $\get{a} \redn u$, it holds that $u\red T$. Indeed, let us consider any $u$ such that $\get{a}\redn u$; since $u$ must be  a deterministic  simply typed $\lambda$-term,  by Corollary \ref{cor:sndt} $u\red T$.

ii) Since $\send{a}$ neutral, by \crtre{} it is enough to show that for all $u$ such that $\send{a} \redn u$,  it holds $u\red T$. Indeed, let us consider any $u$ such that $\send{a}\redn u$; since $u=\get{a}$, by i) we have $u\red T$.

%
%
\end{proof}


\label{section-adequacy}

We can finally prove that all simply typed $\lambda$-terms are reducible.

\begin{theorem}[Adequacy Theorem]\label{Adequacy Theorem}
Suppose that $w: A$ is any simply typed $\lambda$-term, with intuitionistic free variables among $x_1: {A_1},\ldots,x_n:{A_n}$. For all terms $t_1, \ldots, t_n$ such that
 for $i=1,\ldots, n,$ $ t_i\red A_i$
we have $ w[t_1/x_1\cdots
t_n/x_n]\red A$
\end{theorem}

%
%
%
%
%
%
%
%
%
%
%
%
\begin{corollary}[Strong Normalization of $\redn$]\label{cor:snn} For any simply typed $\lambda$-term $t: A$, then $t\in\snn$.
\end{corollary}

\begin{proof}
Let $x_1: {A_1},\ldots,x_n:{A_n}$ be all the intuitionistic free variables of $t$. \crtre{} leads to $x_{i}\red A_{i}$. 
From Th.~\ref{Adequacy Theorem}, we derive $t\red A$. $t\in\snn$ follows from \cruno.
\end{proof}



\begin{theorem}[Strong Normalization of $\lama$] 
\label{SNT}
For any $\lama$ term $t$,  $t\in \sn$.
\end{theorem}
\begin{proof}
Assume $t=\pp{a}{u_1 \p \ldots \p u_m}$ and $u_1= v_{1}\p\dots\p
v_{n}$, $\dots$, $u_{m}=v_{p}\p\dots\p v_{q} $. Define $\rho_{i}$, for
$i\in\{1, \ldots, q\}$, as the ordinal size $h(v_{i})$ of $v_{i}$ with
respect to the reduction relation $\redn$.
We proceed by lexicographic induction on the sequence
$\rho = ( \rho_1, \ldots , \rho_q) $,
which we call the complexity of $t$.
Let $t'$ be any term such that $t\mapsto t'$: to prove the thesis, it is enough to show that $t'\in \sn$. Let $\rho'$ the complexity of $t'$.  We must consider three cases.
\\
\textbf{1.} $t'=\pp{a}{u_1 \p \ldots \p u_{k}'\p \dots \p u_m}$ and in turn 
$u_{k}'=v_i \p \ldots \p v_{r}'\p \dots \p v_j$
with $v_{r}\mapsto v_{r'}$ by contraction of an intuitionistic redex. Then also $v_{r}\redn v_{r}'$. Hence $\rho'$ is lexicographically strictly smaller than $\rho$ and we conclude by the i.h. that $t'\in \sn$.\\
\textbf{2.} $ t'=\pp{a}{u_1' \p \!\shortdots\! \p u_{k}'\p\! \shortdots\! \p u_m'}\;  $ $ u_{k}= \!\shortdots\! \p v_i \p\! \shortdots\!
\p v_j\p\! \shortdots\!\;   $ $ u_{k}'=\!\shortdots\! \p v_i' \p\! \shortdots\! \p v_j'\p
\!\shortdots\!    $  $  v_{i}={\mathcal D}_1[\get{a}\, w_1], \shortdots ,  v_{j}={\mathcal
D}_{n}[\get{a} \, w_n] \qquad\quad   $ $ v_{i}'={\mathcal D}_1[\lan w_1, w\ran],
\shortdots v_{j}={\mathcal D}_{n}[\lan w_{n}, w\ran]$\\where $w$ is
the sequence of messages transmitted by cross reduction. Being for
each $l$\\\centerline{$ \mathcal{D}_l[\get{a}\, w_l]\redn \mathcal{D}_l[(\lambda x\, \lan x,
    w\ran)\, w_l]\redn \mathcal{D}_l[\lan w_{l}, w\ran]$}\\
we obtain $v_{i}\redn^{*} v_{i}', \dots v_{j}\redn^{*}
v_{j}'$. Moreover, for each $l\neq k$, either $u_{l}=u_{l}'$ or \\
$ u_{l}= \!\shortdots\! \p s_i \p\! \shortdots\!
\p s_j\p\! \shortdots\!\quad   $ $ u_{k}'=\!\shortdots\! \p s_i \p\! \shortdots\! \p s_j'\p
\!\shortdots\!    \quad  s_{j}={\mathcal
C}[\send{a} \, r] \redn  s_{j}'={\mathcal C}[\get{a}\, r]$

Hence $\rho'$ is lexicographically strictly smaller than
$\rho$ and by i.h. $t'\in \sn$.\\ 
\textbf{3.} $t'=v_{i_{1}}\p\dots\p v_{i_{k}}$. As $v_{i_{1}},\dots, v_{i_{k}}$ all belong to $\sn$, we easily obtain   $t'\in \sn$.\end{proof}

\section{Computing with $\lama$}\label{sec:computing}


We illustrate the expressive power of $\lama$  
with examples of parallel programs from~\cite{loogen2011}. 

Reductions are performed
according to the principles (a)--(c) below, that make programming in $\lama$ as efficient and
deterministic as possible. Indeed, as put by Harper~\cite{Harper}, whereas concurrency is concerned with nondeterministic composition of programs,  parallelism is concerned with asymptotic \emph{efficiency} of programs with \emph{deterministic behavior}.


(a) Messages should be normalized before being sent.  This property is
fundamental for efficient parallel computation. Indeed, for instance,
a repeatedly forwarded message can be duplicated many times inside a
process network;
hence if the message is not normal, each process may have to normalize it,
wasting resources. This implies in particular that message
transmission cannot be triggered according to a call-by-value
strategy.


(b) The senders and the receiver of the message should be normalized
before the communication. Indeed, the communication behaviour depends
on the syntactic shape of a process: it is the rightmost occurrence of
the channel in a $\lama$ term to determine what message is sent or
received. Hence, a term might send a message, whereas its normal form
another. For example $(\lambda x\, y\, x\, (a n)) (am)$ would
transmit $m$, while its normal form $y\, (am)\, (a n)$ would transmit
$n$. We want to avoid this kind of unpredictable behaviour due to the non-confluence of the calculus and make
clear what channel is supposed to be the active one, which sends or
receives a message.



 
 (c) When no communication nor intuitionistic reduction is possible, we extract the desired result of the computation by a suitable simplification reduction.

Informally, our normalization strategy consists in iterating the
basic reduction relation $\succ$ defined below, that takes a term
$t=\pp{a} {u_{1}\p \ldots\p u_m }$ and performs the following
operations:

1) We select all the threads of $u_{1}, \dots, u_{m}$ that will send
or receive the next message and we let them complete their internal
computations.

2) We let the selected threads transmit their message.


3) While executing 1) and 2), to avoid inactivity we let the other
threads perform some other independent calculations that may be
carried out in parallel. 

4) If the previous operations are not possible, we extract the results, if
any.




\begin{definition}[Reduction Strategy $\succ$]\label{defi-redstrategy}
  Let $t=\pp{a}{ u_1 \p \shortdots \p u_m }$ be a $\lama$ term.  We write
$t\succ t'$
whenever $t'$ has been obtained from $t$ by applying on of the following:
\begin{enumerate}
\item We select a receiver $u_i$ for the next communication. We normalize,
among the threads that contain $a$, those that are rightmost in $u_i$
or in a process outlinked to $u_i$. If now it is possible, we apply a
cross reduction
\begin{footnotesize}
  \[\pp{a}{ u_1 \p \ldots \p u_i \p \ldots \p u_m }\mapsto \pp{a}{ u_1
      \p \ldots \p u_i'\p \ldots \p u_m}\]
\end{footnotesize}followed by some
intuitionistic reductions.
\item Provided that by 1.\ 
we can only obtain the trivial reduction $ t \mapsto^* t $, we apply, if possible, a simplification reduction.
We then normalize the remaining simply typed $\lam$-terms. 
\end{enumerate}
\end{definition}
We can reduce every $\lama$-term in normal form just by
iterating the reduction relation $\succ$. Indeed, if a communication
is possible, by 1.\ we can select a suitable receiver and apply a 
cross reduction. If no communication is possible but a simplification
can be done, 2.\ applies and we can simplify the term. Otherwise, we
can just normalize the simply typed $\lam$-terms by 1.\ or 2.

This normalization strategy leaves some room for non-determinism: it
prescribes when a communication reduction should be fired, but does
not select a process out of those that can potentially receive
messages, thus leaving a number of possible ways of actually
performing the communication. To limit this non-determinism, before the beginning of the computation we select one process and we impose that only the selected process can receive messages. Immediately after the reception of the messages, we select the next process to the right, if any, or the first process  from left
otherwise. Another source of non-determinism is due to cross
reductions of the form\[\ppp{a}{\ax}{ (u_{m_{1}} \p \ldots \p u_{n_{1}}) \p
\ldots \p (u_{m_{p}} \p \ldots \p u_{n_{p}}) } \mapsto u_{j_1} \p
\ldots \p u_{j_q}\] In this case, we impose that $u_{j_1} \p \ldots \p
u_{j_q}$ results by selecting from each $(u_{m_{i}} \p \ldots \p
u_{n_{i}})$, with $1\leq i\leq p$, the leftmost thread not containing
$a$.

Before presenting the parallel programs for computing $\pi$ and the all-pair-shortest-path problem,
 we show that $\lama$ is more expressive than simply typed $\lambda$ calculus.
\begin{example}[Parallel OR] 
Berry's sequentiality theorem~\cite{barendregt84}
implies that there is no simply typed
$\lambda$-term $O: \mathsf{Bool} \IMPL \mathsf{Bool}
\IMPL\mathsf{Bool} $ such that $O\, \mathbf{ff}\, \mathbf{ff}
\mapsto^* \mathbf{ff}$, $O\, u\, \mathbf{tt} \mapsto^* \mathbf{tt}$,
$O\, \mathbf{tt} \, u \mapsto^* \mathbf{tt}$ for \emph{every}
$\lam$-term $u$, where $\mathbf{tt}, \mathbf{ff}$ are the boolean
constants. As a consequence, there cannot be a $\lam$-term $\mathsf{O}$
such that $\mathsf{O} [ \mathbf{ff}/x][ \mathbf{ff}/y] \mapsto^*
\mathbf{ff}$, $\mathsf{O} [u/x][ \mathbf{tt}/y] \mapsto^*
\mathbf{tt}$, $\mathsf{O} [\mathbf{tt}/x][u /y] \mapsto^* \mathbf{tt}$
for every $\lam$-term $u$. 
To implement a $\lama$-term with the above property  we process the
two inputs in parallel. If both inputs evaluate to $\mathbf{ff}$,
though, at least one process needs to have all the information in
order to output the result $\mathbf{ff}$. Hence the simple
topology\tikzstyle{proc}=[circle, minimum size=2mm, inner sep=0pt,
draw]
\begin{tikzpicture}[node distance=1.5cm,auto,>=latex', scale=0.6]
\node [proc] (1) at (0,0) {}; \node [proc] (2) at (2,0) {}; \path[->]
(1) edge [thick, bend left=10] (2);\end{tikzpicture} is enough.
Proc.~\ref{proc:topology} extracts from this graph the axiom
schema $ (\mathbb{A}\rightarrow \mathbb{A}\land \bot) \vel
(\mathbb{B}\impl \mathbb{B} \et \mathbb{A}) $ with
reduction\begin{footnotesize}
$$ \pp{a}{\mathcal{C}[\send{a} w ]\p ({\mathcal D}_1[\get{a}v_1] \p \!\shortdots\!\! \p
    {\mathcal D}_n[\get{a}v_n] )} \; \mapsto\; \pp{a} {\mathcal{C}[\get{a} w ]\p
({\mathcal D}_1[\lan v_1 , w\ran ] \p \! \shortdots\!\! \p {\mathcal
D}_n[\lan v_n, w\ran ] )} $$
\end{footnotesize}We add to $\lama$ the boolean type, $\mathbf{tt},
\mathbf{ff}$ and an $\mathsf{if} \_\mathsf{then}\_ \mathsf{else}\_ $
construct~\cite{Girard}.  We define in $\lama$

\centerline{{\small $\mathsf{O}  := \pp{a}{\mathsf{if} \; x \; \mathsf{then} \;
      \mathbf{tt} \; \mathsf{else} \; {\send{a}\, \mathbf{ff}\, \pi_{0} }
      \p \mathsf{if} \; y \; \mathsf{then} \; \mathbf{tt} \;
        \mathsf{else} \; \get{a} \, (\lam x^{\fal} \, x) \, \pi_1 } $}}\smallskip

\noindent where we assume $\send{a}: \mathsf{Bool} \IMPL
\mathsf{Bool}\land \bot$ in the first process and that $\get{a}: \ver \impl
\ver \et \mathsf{Bool}$ in the second one. Now, on one hand
 \begin{footnotesize}$$\mathsf{O} [ u/x][ \mathbf{tt}/y] =
\pp{a}{\mathsf{if} \, u \, \mathsf{then} \, \mathbf{tt} \,
\mathsf{else} \,{\send{a}\, \mathbf{ff}\pi_{0} }\parallel \mathsf{if}
\,\mathbf{tt} \, \mathsf{then} \, \mathbf{tt} \, \mathsf{else} \, \get{a} \,
(\lam x^{\fal} \, x) \, \pi_1 } \mapsto^{*} \pp{a} {\mathsf{if} \, u
\, \mathsf{then} \, \mathbf{tt} \, \mathsf{else} \, \send{a}\,
\mathbf{ff}\pi_{0} \parallel \mathbf{tt}} \, \mapsto \,
\mathbf{tt}$$\end{footnotesize}and symmetrically $\mathsf{O}[
\mathbf{tt}/x][u/y]\mapsto^{*} \, \mathbf{tt} $.  On the other hand,
$ \mathsf{O} [ \mathbf{ff}/x][
\mathbf{ff}/y]\mapsto^{*} \mathbf{ff}$.
\end{example}

\begin{example}[A Parallel Program for Computing  $\pi$]

We implement in $\lama$ a parallel program for computing arbitrary
precise approximations of $\pi$. As is well known, $\pi$ can be computed
as follows:
$\pi=\lim_{l\to\infty }
\frac{1}{l} \sum_{i=1}^{l} f(\frac{i-\frac{1}{2}}{l})$ where 
  $f(x)=\frac{4}{1+x^{2}}$.
For any given $l$, instead of
calculating sequentially the whole sum $ \sum_{i=1}^{l}
f(\frac{i-1/2}{l})$ and then dividing it by $l$, it is more efficient
to distribute different parts of the sum to $p$ parallel
processes. When the processes terminate, they can send the results
$m_{1}, \ldots, m_{p}$, as shown below, to a process which computes
the final result $\frac{1}{l}(m_{1}+ \ldots + m_{p})$ (see
also~\cite{loogen2011}).
%
%
The graph in the following figure, 
in which we omit reflexive edges, is encoded by 
the axiom $\ax = (\mathbb{A}_{1}\rightarrow \mathbb{A}_{1}\land \bot)\lor \ldots \lor
(\mathbb{A}_{p}\rightarrow \mathbb{A}_{p}\land \bot) \lor \mathbb{B} \rightarrow ( \mathbb{B} \land
\mathbb{A}_{1} \land \ldots\land \mathbb{A}_{p} )$.
\tikzstyle{proc}=[circle, minimum size=1mm, inner sep=0pt, draw]
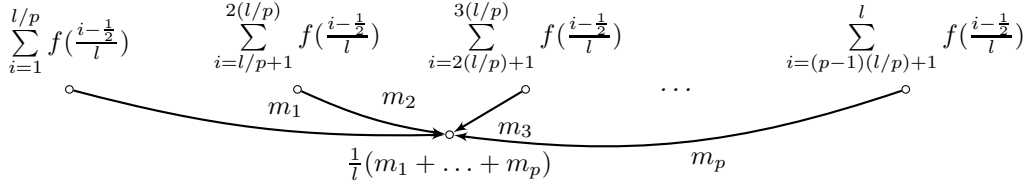
  \begin{tikzpicture}[node distance=1.3cm,auto,>=latex', scale=0.2]

    \node [proc, label={$\sum\limits _{i=1}^{l/p}
f(\frac{i-\frac{1}{2}}{l})$}] (4) at (0,0) {};

\node [proc, label={$\sum\limits _{i=l/p+1} ^{2(l/p)}
f(\frac{i-\frac{1}{2}}{l})$}] (5) at (15,0) {};

\node [proc, label={
$\sum\limits_{i=2(l/p)+1}^{3(l/p)}f(\frac{i-\frac{1}{2}}{l})$}] (6) at
(30,0) {};


\node [] (1) at (40,0) {\ldots}; 

\node [proc, label={$\sum\limits
_{i=(p-1)(l/p)+1}^{l} f(\frac{i-\frac{1}{2}}{l})$}] (2) at (55,0) {};

\node [proc, label={below: $\frac{1}{l}(m_{1}+ \ldots + m_{p})$}] (3)
at (25,-3) {};

\path[->] (4) edge [thick, bend right=8] node {$m_1$} (3);

\path[->] (5) edge [thick, bend right=8] node {$m_2$} (3);

\path[->] (6) edge [thick] node {$m_3$} (3);

\path[->] (2) edge [thick, bend left=13] node {$m_p$} (3);

\end{tikzpicture}\\
To write the program, we add to
$\lama$ types and constants for rational numbers, together with
function constants $\mathsf{f}_{1}, \dots, \mathsf{f}_{p}, \mathsf{sum}$ such that
\begin{small}
$$\mathsf{f}_{k}\, l \mapsto^{*}
\sum_{i=(k-1)l/p+1}^{(kl/p)} f(\frac{i-\frac{1}{2}}{l}) \quad
\mbox{and} \quad \mathsf{sum}\, \lan n_{1}, \ldots, n_{i}\ran\,
l\mapsto^{*} \frac{1}{l}(n_{1}+\ldots+ n_{i})$$
\end{small}thus computing respectively
the $p$ partial sums and the final result.
The $\lama$ term that takes as input the length $l$ of the
summation to be carried out and yields the corresponding approximation
of $\pi$ is:$\quad  \pp{a}{\ \send{a}( {\mathsf{f}_{1}\, l)}\pi_{0} \p \ldots \p
{\send{a}(\mathsf{f}_{p}\, l)}\pi_{0}\p \mathsf{sum}\, (\get{a}\, (\lambda
x^{\bot}\, x)\, \pi_{1})\, l\ }\; $. By instantiating   $\mathbb{A}_{1}, \ldots ,
\mathbb{A}_{p}$ in the extracted $\ax$ with the type $\mathsf{Q}$ of rational numbers, we type all displayed occurrences of $a$ in the
first $p$ threads by $\mathsf{Q} \impl\mathsf{Q}\land \bot$, and the last by $\ver
\rightarrow ( \ver \land \mathsf{Q} \land \ldots\land \mathsf{Q} )$.
%
%
%
%
%
%
Given any multiple $n$ of $p$, we have\\{\footnotesize$(\pp{a}{\ \send{a}(	{\mathsf{f}_{1}\, l) \,\pi_{0}} \!\p \!\shortdots\!\! \p\! {\send{a}(\mathsf{f}_{p}\, l)\,\pi_{0}}\!\p \! \mathsf{sum}\, (\get{a}\, (\lambda x^{\bot}\, x)\, \pi_{1})\, l\  })[n/l] = \pp{a}{\ \send{a}(	{\mathsf{f}_{1}\, n) \,\pi_{0}} \!\p \!\shortdots\!\! \p\! {\send{a}(\mathsf{f}_{p}\, n)\,\pi_{0}} \!\p  \mathsf{sum}\, (\get{a}\, (\lambda x^{\bot}\, x)\, \pi_{1})\, n\ }$}\begin{footnotesize}$$\mapsto\; \pp{a}{\ \send{a}( {\sum_{i=1}^{(n/p)} f(\frac{i-\frac{1}{2}}{l}))\,\pi_{0}} \p \ldots \p
{\send{a}(\sum_{i=(p-1)n/p+1}^{n} f(\frac{i-\frac{1}{2}}{l}))\,\pi_{0}}\p
\mathsf{sum}\, (\get{a}\, (\lambda x^{\bot}\, x)\, \pi_{1})\, n\ }$$
  $$\mapsto^{*} \; \mathsf{sum}\, (\lan\sum_{i=1}^{(n/p)}
    f(\frac{i-\frac{1}{2}}{l}))), \ldots, \sum_{i=(p-1)n/p+1}^{n}
f(\frac{i-\frac{1}{2}}{l}))\ran\, n\ \quad \mapsto^{*} \;\frac{1}{n}
\sum_{i=1}^{n} f(\frac{i-\frac{1}{2}}{n})$$\end{footnotesize}\end{example}

\begin{example}[A Parallel Floyd--Warshall Algorithm]
\label{ex:floyd-warshall}
 
We define a $\lama$ term that implements a parallel version of the
Floyd--Warshall algorithm.
The algorithm takes as input a directed 
graph and outputs a matrix containing the length of the
shortest path between each pair of nodes.
%
%
%
%
Formally, the input graph is coded as a matrix $I(0)$ and the nodes of
the graphs are labeled as $1, 2, \ldots, n$. Then the sequential
Floyd--Warshall algorithm computes a sequence of matrixes $I(1),
\ldots, I(n)$, representing closer and closer approximations of the
desired output matrix.  In particular, the entry $(i,j)$ of $I(k)$ is
the length of the shortest path connecting $i$ and $j$ such that every
node of the path, except for the endpoints, is among the nodes $1, 2,
\ldots k$.  Now, each matrix $I(k)$ can be easily computed from
$I(k-1)$.  The idea is that passing through the node $k$ might be
better than passing only through the first $k-1$ nodes, or not.
Hence in order to compute $I(k+1)$ from $I(k)$, one only needs to
evaluate the right-hand side of the following equation:
\begin{equation}
\label{recursion} I_{i, j}(k)= \mathsf{min}(I_{i,j}(k-1), I_{i,
k}(k-1) + I_{k, j}(k-1))
\end{equation} Indeed, $I_{i,j}(k-1)$ represents the shortest
path between $i, j$ passing only through the first $k-1$ nodes, while
$I_{i,k}(k-1) + I_{k, j}(k-1)$ computes the shortest path between $i, j$ that
passes through $k$ and the first $k-1$ nodes.
To speed-up this computation, we create a parallel process for each row $I(k)$ of the matrix and put the process in charge
of computing the row.  We say that the $i$-th process has to compute
the $i$-th row $I_{i}(k)$. Which information does it need?
Actually, only two rows: $I_{i}(k-1)$ and $I_{k}(k-1)$. Hence at
each round $k$ the process-$i$ only lacks the row $I_{k}(k-1)$ to
perform its computation. This row can be communicated to process-$i$
by the process-$k$, in charge to compute that row. These
considerations lead to a well-known parallel algorithm.

\subsubsection*{The ring Floyd--Warshall algorithm}

1. Take the input $n\times n$-matrix $I(0)$ and distribute the $i$-th
row $I_{i}(0)$ to process-$i$. Organize the $n$ processes in a ring
structure such as the following (see~\cite{loogen2011}), omitting
reflexive edges:\begin{center}
  \tikzstyle{proc}=[circle,minimum size=2mm, inner sep=0pt, draw]
  
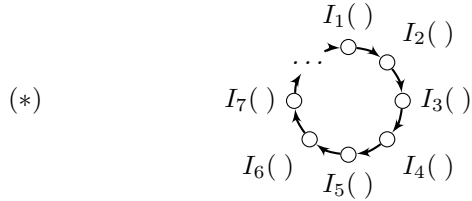
\hspace{-140pt}\begin{tikzpicture}[node
    distance=1.3cm,auto,>=latex', scale=0.17]

\node [] (0) at (-25,0)
    {$(*)$};

\node [] (0) at (-3,3) {\ldots};

\node [proc, label={above: $I_1(\; )$}] (1) at (0,4.2) {};
    
\node [proc, label={above right: $I_2(\; )$}] (2) at (3,3) {};

\node [proc, label={right: $I_3(\; )$}] (3) at (4.2,0) {};

\node [proc, label={below right: $I_4(\; )$}] (4) at (3,-3) {};

\node [proc, label={below: $I_5(\; )$}] (5) at (0,-4.2) {};

\node [proc, label={below left: $I_6(\; )$}] (6) at (-3,-3) {};

\node [proc, label={left: $I_7(\; )$}] (7) at (-4.2,0) {};

\path[->] (1) edge [thick, bend left=15] (2);

\path[->] (2) edge [thick, bend left=15] (3);

\path[->] (3) edge [thick, bend left=15] (4);

\path[->] (4) edge [thick, bend left=15] (5);

\path[->] (5) edge [thick, bend left=15] (6);

\path[->] (6) edge [thick, bend left=15] (7);

\path[->] (7) edge [thick, bend left=15] (0);

\path[->] (0) edge [thick, bend left=15] (1);
  \end{tikzpicture}
\end{center}
\vspace{-0.1cm}
\noindent 2. For $k=1$ to $n$, starting from process-$k$, let all the processes
forward the row $I_{k}(k-1)$ to their successors in the ring, until
the process $k+1$ receives again the same row. After the row has
circulated, let the processes compute in parallel the rows of the
matrix $I(k)$.

\noindent 3. Let $I(n)$ be the output.
\smallskip

\noindent 
In this algorithm, at any stage
$k$,  communicating the required row $I_{k}(k-1)$ through the ring
structure requires less time compared to computing
each row of the matrix, therefore the overhead of the communication is
compensated by the speed-up in the matrix computation.

\subsubsection*{A  $\lama$ program for the Floyd--Warshall algorithm}

To write a $\lama$ program that computes the ring
Floyd--Warshall Algorithm, we add integers to $\lama$,
as usual, and we denote with $A$ the function type
corresponding to the rows of the matrixes $I(k)$ computed by the
algorithm. The expression $\mathrm{I}_x(y)$ represents the function
computing the $x$\textsuperscript{th} line of the matrix at the
$y$\textsuperscript{th} stage of the algorithm.  We assume that the value of ${I}_x(y)$ 
contains information about $x$ and $y$.
We also add the
constant function $f:A \et A \impl A $ such that:
$f\lan \mathrm{I}_i(k-1) , \mathrm{I}_{k}(k-1) \ran = \mathrm{I}_i(k) $ and
$f\lan \mathrm{I}_z(l) , \mathrm{I}_{m}(n) \ran = \mathrm{I}_z(l)$.

The first equation implements the calculations needed for equation
(\ref{recursion}). The second comes into play at the end of each
iteration of the algorithm, when the process that receives twice the
same row discards it, starts sending its own row, and
begins the next iteration of the algorithm.



As handy notation, for any three terms $u, v, s$ we define $(u,v)^1
s$ as $u(v s)$, and $(u,v)^{n+1} s$ for $n>0$ as $u(v((u,v)^n
s))$. Moreover, for any two terms $u, s$ we define $(u,\pi_i)^1 s$ as
$u(s\,\pi_i)$, and $(u,\pi_i)^{n+1} s$ for $n>0$ as $u(((u,\pi_i)^n s)
\pi_i)$. Intuitively, $ (u,v)^{n} s$ represents what we
obtain if we take a term $s$ and then apply alternately $v$ and $u$
to the term resulting from the last operation.  We obtain,
ultimately, a term of the form $u(v (u(v(\ldots u(vs) \ldots))))$ in
which and $u$ and $v$ occur $n$ times each.  The notation
$(u,\pi_i)^{n} s$ is analogous, but instead of applying $v$ we apply
the projection $\pi_i$.
The $n$ processes that run in
parallel during the execution of the algorithm are ($1< i < n$): 
\\ 
\emph{Process}
$p_1$:  $\quad (f,\get{a})^{n} \, I_1(0) \p ((\send{a}, \pi_1)^{n}\,\get{a}
I_1(0) )\pi_0 \p (\send{a} I_1(0) ) \pi_0 $

\noindent \emph{Process} $p_{i}$: 
\begin{small}
  $(f,\get{a})^{n+1} \, I_i(0) \p ((\send{a}, \pi_1)^{n+1-i}\, (\get{a}, \pi_{1})^{i}\,\get{a}\,I_i(0) )
    \pi_0 \p (\send{a} \, (f,\get{a})^{i} \, I_i(0) ) \pi_0 \p ((\send{a}, \pi_1)^{i-1}\,
    \get{a} I_i(0)  ) \pi_0 $ 
\end{small}\emph{Process} $p_n$:
  $ (f,\get{a})^{n+1} \, I_n(0)
    \p ( \send{a} \, ((f,\get{a})^{n} \, I_n(0)) ) \pi_0 \p (( \send{a} ,
    \pi_1)^{n-1}\, \get{a} I_n(0) )\pi_0 $ 

For an intuitive reading of the parts of these terms,
consider the notation $(f,\get{a})^{m}\, t $. This notation represents a
term of the form $f(\get{a} ( \ldots f(\get{a}t) \ldots ))$. Only one operation
can be immediately performed with a term like this: using the
innermost application of $a$ and consume $at$ to receive a message. The received message will
be the argument of the innermost $f$. The value that $f$ computes, in
turn, will be the argument of the next communication
channel. Terms of this form alternate two phases: one in
which they receive, and one in which they apply $f$ to
the received message.   The terms
$(\send{a},\pi_1)^{m}\, t $ have the form $\send{a}( ( \ldots \send{a}( t \pi_1) \ldots
)\pi_1)$. These terms project, send and receive; and then start
over. The projections are used to select messages from the tuple of received messages. The selected messages are not used, but forwarded to another process.

Each process $p_{i}$ ($i>1$) is a parallel composition
of four threads: the first and the third compute
the row $I_{i}(k)$, and the second and the fourth 
send and forward rows. The process $p_1$ behaves in the same
way, but has three threads: the first computes the rows $I_{1}(k)$,
the second receives and forwards rows, and the third sends its own
row. The term implementing the 
algorithm is $\pp{a} { \, p_{1} \p \ldots \p p_{n} }$. The axiom
extracted from the ring structure above by
Proc.~\ref{proc:topology} is
$\ax=(\mathbb{A}_{1} \rightarrow \mathbb{A}_{1} \et
    \mathbb{A}_{n})\lor (\mathbb{A}_{2} \rightarrow \mathbb{A}_{2} \et \mathbb{A}_{1} ) \lor\ldots \lor
    (\mathbb{A}_{n-1} \rightarrow \mathbb{A}_{n-1} \et \mathbb{A}_{n-2} ) \lor (\mathbb{A}_{n} \rightarrow
    \mathbb{A}_{n} \et \mathbb{A}_{n-1}) $.
We instantiate $\mathbb{A}_{1} , \ldots , \mathbb{A}_{n}$ in $\ax$  with $A$ and type $a$ by  $A \impl A \et A $.  

We present some steps of the execution of the 
instance \mbox{$\pp{a}{\, p_1 \p
p_2  \p p_3}  $}. 
According to our normalization principles,
we normalize the threads sending messages right before they
communicate; then, we normalize the rightmost threads receiving
the messages. All other $\IL$ reductions are performed
in-between communications. Finally, we make sure to contract all
intuitionistic redexes before the beginning of a new iteration of the
algorithm.
We start with:
\begin{footnotesize}
  \begin{align*} \mapsto^* 
  & \po{a}  \big(   f(\get{a}(f(\get{a}(f(\get{a}  I_1(0)))))) \p
      (\send{a}((\send{a}((\send{a}((\get{a}   I_1(0) ) \pi_1))\pi_1))\pi_1)) \pi_0 \p
      (\send{a} I_1(0)) \pi_0   \big)     \\
    & \p
      \big(   f(\get{a}(f(\get{a}(f(\get{a}(f(\get{a}  I_2(0)))))))) \p  (\send{a}((\send{a}((\get{a}((\get{a}((\get{a}
       I_2(0))
      \pi_1))\pi_1))\pi_1))\pi_1)) \pi_0
  \\ & \quad
      \p  (\send{a}(f(\get{a}( f  (\get{a} I_2(0))))) ) \pi_0 \p (\send{a}((\get{a}   I_2(0) ) \pi_1)) \pi_0
      \big)     \\ & \p \big(   f(\get{a}(f(\get{a}(f(\get{a}(f(\get{a}  I_3(0))))))))
                     \p (\send{a}(f(\get{a}( f(\get{a}(f  (\get{a}  I_3(0)))))))) \pi_0
    \p (\send{a}((\send{a}((\get{a}  I_3(0) ) \pi_1))\pi_1)) \pi_0 \big) \pc
\end{align*}
\end{footnotesize}We first transmit the value $I_1(0)$ from the
first process to the second process and we move the focus to the next
term. As an aid to the reader, we display
between $ ^{\star} \; ^{\star}$ the occurrences of the message that are involved in the reduction:\begin{footnotesize}\begin{align*} & \po{a} \big( f(\get{a}(f(\get{a}(f(\get{a}
I_1(0)))))) \p (\send{a}((\send{a}((\send{a}((\get{a}  I_1(0) ) \pi_1))\pi_1))\pi_1)) \pi_0 \p
( ^{\star} \get{a} I_1(0) ^{\star}  ) \pi_0 \big)   \\ 
& \p  \big( f(\get{a}(f(\get{a}(f(\get{a}(f \lan I_2(0), ^{\star} I_1(0) ^{\star}   \ran)))))) \p
 (\send{a}((\send{a}((\get{a}((\get{a}(\lan  I_2(0), ^{\star} I_1(0) ^{\star} \ran\pi_1))\pi_1))\pi_1))\pi_1) )
  \pi_0    
   \\ & \quad 
  \p (\send{a}(f(\get{a}(f \lan I_2(0),^{\star} I_1(0) ^{\star} \ran)))) \pi_0 \p (\send{a}(\lan
    I_2(0), ^{\star} I_1(0) ^{\star}\ran \pi_{1}  )) \pi_0 \big)
\\ 
& \p
      \big(  f(\get{a}(f(\get{a}(f(\get{a}(f(\get{a} I_3(0))))))))
(\send{a}(f(\get{a}(
      f(\get{a}(f (\get{a} I_3(0)))))))) \pi_0
             \p (\send{a}((\send{a}((\get{a} (\lan 0, I_3(0)\ran \pi_1))
       \pi_1))\pi_1)) \pi_0
       \big) \pc \end{align*}\end{footnotesize}
\begin{remark}
\label{memory}
  $I_2(0)$ is not destroyed by the communication but saved
  using the memorization mechanism. This term is needed, indeed, by
  the function $f$ in order to compute $I_2(1)$.
\end{remark}We normalize the receiver of the previous communication and keep
normalizing the other redexes in parallel, thus the rightmost thread
of the second process become ready to forward the value $I_1(0)$ to
the third thread:
\begin{footnotesize}\begin{align*}
\mapsto^* & \po{a}  \big( f(\get{a}(f(\get{a}(f(\get{a} I_1(0)))))) \p
(\send{a}((\send{a}((\send{a}((\get{a}  I_1(0) ) \pi_1))\pi_1))\pi_1)) \pi_0 \p
(\get{a}(I_1(0)) \big)) \pi_0   \\ 
& \p \big( f(\get{a}(f(\get{a}(f(\get{a} (f \lan I_2(0), I_1(0) \ran)))))) \p
(\send{a}((\send{a}((\get{a}((\get{a} I_1(0))\pi_1))\pi_1))\pi_1)) \pi_0 \\ & \quad\p (\send{a}(f(\get{a} (f \lan
  I_2(0), I_1(0) \ran)))) \pi_0 \p (\; ^{\star}\send{a}I_1(0)^{\star}\;  ) \pi_0 \big)
\\ 
& \p 
      \big(   f(\get{a}(f(\get{a}(f(\get{a}(f(\get{a} I_3(0))))))))
  \p (\send{a}(f(\get{a}(
      f(\get{a}(f (\get{a} I_3(0)))))))) \pi_0
                           \p (\send{a}((\send{a}((\get{a} I_3(0)) \pi_1))\pi_1)) \pi_0
\big) \pc    \end{align*}\end{footnotesize}The third thread
receives $I_1(0)$ and the computation continues: $\quad \mapsto  \dots$\\
At the end of the first of the three cycles, the second process receives $I_1(0)$ again:
\begin{footnotesize}\begin{align*}
 \mapsto^* & \po{a} \big( f(\get{a}(f(\get{a} f \lan I_{1}(0), I_1(0)\ran))) \p
(\send{a}((\send{a}(( \, ^{\star} \get{a} I_1(0) ^{\star}\,)\pi_1))\pi_1) ) \pi_0 \p
\lan I_1(0) , I_1(0)\ran  \pi_0 \big)    \\ 
& \p  
\big( f(\get{a}(f(\get{a}(f \lan  I_2(0), \, ^{\star}I_1(0) ^{\star}\,  \ran ))))) \p (
 \send{a}((\send{a}((\get{a}( \lan I_1(0), \,  ^{\star} I_{1}(0) ^{\star} \,  \ran\pi_{1})\pi_1))\pi_1)) \pi_0 
 \\& \quad
   \p (
  \send{a}(f\lan  I_2(0) , \,  ^{\star}I_1(0) ^{\star} \, \ran ))) \pi_0 \p \lan
                                                                                   I_1(0),\,  ^{\star} I_1(0) ^{\star}\,  \ran \pi_0\big)  
\\
                      & \p  \big(
     f(\get{a}(f(\get{a}(f(\get{a} f \lan I_3(0), I_{1}(0)\ran))) )) 
     \p ( \send{a}(f(\get{a}(
                        f(\get{a} f \lan I_3(0), I_{1}(0)\ran))))) \pi_0
                      \p (\send{a}((\get{a} I_1(0))\pi_1)) \pi_0
                      \big)
  \pc \end{align*}\end{footnotesize}The reception of $I_1(0)$
exhausts all forwarding channels of the second process  in the
rightmost thread  and triggers the communication of the 
$I_2(1)$ just computed, thus starting the second cycle:\begin{footnotesize}
\begin{align*}
 \mapsto^* & \po{a} \big( f(\get{a}(f(\get{a} I_1(1)))) \p
( \send{a}((\send{a}((\get{a} I_1(0))\pi_1))\pi_1) ) \pi_0\p (\lan I_1(0) , I_1(0)\ran )
             \pi_0 \big)  
\\ 
& \p  \big( f(\get{a}(f(\get{a} I_2(1) ))) \p
 ( \send{a}((\send{a}((\get{a}I_1(0))\pi_1))\pi_1)) \pi_0 \p (\, ^{\star} \send{a} I_2(1)^{\star}\,   ) \pi_0
  \p \lan I_1(0), I_1(0) \ran  \pi_0 \big)    \\ 
           & \p \big(
 f(\get{a}(f(\get{a}(f(\get{a} I_3(1)))))) 
                      \p
  (\send{a}(f(\get{a}( f(\get{a}
             I_3(1))))) ) \pi_0
  \p ( \send{a}((\get{a} I_1(0))\pi_1)) \pi_0
  \big)\pc
 \quad  \mapsto  \dots 
\end{align*}\end{footnotesize}The reduction continues in the same
fashion for two other cycles, until the third process does not contain
any occurrence of $a$ anymore:
\begin{footnotesize}\begin{align*}
                      \mapsto^* & \po{a} \big(
                                  f\lan I_1(2),  I_3(2)\ran \p ( \send{a}
            I_3(2))  \pi_0 \p \lan I_1(0) ,
                                  I_1(0)\ran \pi_0
                                  \big) 
                                  \\ & \p  \big( f \lan  I_2(2),  I_3(2)\ran \p
                             ( \get{a} I_3(2) )  \pi_0 \p
\lan I_2(1), I_2(1)\ran \pi_0 \p \lan I_1(0), I_1(0) \ran \pi_0 \big) \\ & \p
 \big( f\lan I_3(2), I_3(2) \ran 
                                                                           \p
                                                                           \lan I_3(2), I_3(2) \ran  \pi_0 \p \lan I_2(1), I_2(1)
\ran  \pi_0 \big) \pc \end{align*}
\end{footnotesize}We then compute all occurrences of $f$ in parallel:
\begin{footnotesize}\begin{align*}
                      \mapsto^* & \po{a} \big(
                                  I_1(3) \p ( \send{a} I_3(2))  \pi_0 \p \lan I_1(0) ,
                                  I_1(0)\ran \pi_0
                                  \big) 
                                   \\ & 
                                  \p  \big( I_2(3) \p  ( \get{a} I_3(2))  \pi_0 \p
\lan I_2(1), I_2(1)\ran \pi_0 \p \lan I_1(0), I_1(0) \ran  \pi_0\big) \\ & \p
                                                                           \big(  I_3(3) \p 
\lan I_3(2), I_3(2) \ran \pi_0 \p \lan I_2(1), I_2(1)
\ran \pi_0  \big) \pc  \end{align*}
\end{footnotesize}We use a simplification reduction to
remove the intermediate results no more needed and only keep
the threads containing the result of the computation: $\;\mapsto^* I_1(3) \p I_2(3) \p
I_3(3)\,$.

\end{example}
\subsection*{Conclusions and Related Work}\label{sec:relwork}
We introduced $\lama$, a simple and yet expressive
typed parallel $\lambda$-calculus based on (suitable fragments of) 
propositional intermediate logics.  
%
$\lama$ is based on and greatly simplifies the calculi in~\cite{ACGTCS,ACG2018,ACGlics2017}.
The basic communication reductions in the
 calculi $\lambda_{CL}$~\cite{ACG2018} and $\lambda_G$~\cite{ACGlics2017} 
implement {\em particular} $\lama$ communications: the message passing
mechanisms based on classical logic 
and on G\"odel logic, 
respectively.  $\lambda_{CL}$ and $\lambda_G$ also contain 
additional reductions ({\em permutations} and {\em full cross}), needed
for proving weak normalization and the subformula property.
These calculi were generalized in~\cite{ACGTCS}, 
that introduces a Curry--Howard correspondence for all propositional intermediate logics characterized by classical disjunctive tautologies, 
the axioms for classical and G\"odel logic being particular cases. The reductions in~\cite{ACGTCS} are the same as 
for $\lambda_{CL}$ and $\lambda_G$ but their 
activation procedure is based on transmitting values and thus it is logic independent.
The simple(r) $\lama$ reduction allows us to encode interesting
parallel programs in Eden -- an extension of Haskell~\cite{loogen2011,Eden}, to prove 
strong normalization and to specify process networks in an automated way.
The price to pay is the lack of the subformula property 
and the fact that well-typed $\lama$-terms might contain deadlock.
In a model of parallel computation the latter is an unwanted feature,
 especially in absence\footnote{Note 
that as the number of $\lama$ communications is not type-determined, the system is compatible with the addition of a recursion operator.} of recursion or fixed points. 
Unlike $\lama$,
the typed versions of the $\pi$-calculus, starting with the seminal work in \cite{CairesPfenning},
are usually deadlock free, e.g., \cite{Wadler2012,TCP2013,Montesi}.

\bibliographystyle{plain}
\bibliography{bib_floyd-warshall}

\begin{thebibliography}{10}

\bibitem{Amadio}
R.M. Amadio.
\newblock On stratified regions.
\newblock In {\em Programming Languages and Systems, 7th Asian Symposium
  ({APLAS} 2009)}, number 5904 in LNCS, pages 210--225, 2009.

\bibitem{ACGlics2017}
F.~Aschieri, A.~Ciabattoni, and F.A. Genco.
\newblock G{\"{o}}del logic: From natural deduction to parallel computation.
\newblock In {\em {LICS} 2017}, pages 1--12, 2017.

\bibitem{ACG2018}
F.~Aschieri, A.~Ciabattoni, and F.A. Genco.
\newblock Classical proofs as parallel programs.
\newblock In {\em {\rm Proceedings of Gandalf 2018}}, volume 277 of {\em
  Electronic Proceedings in Theoretical Computer Science}, pages 43--57, 2018.

\bibitem{ACGTCS}
F.~Aschieri, A.~Ciabattoni, and F.A. Genco.
\newblock On the concurrent computational content of intermediate logics.
\newblock {\em Theoretical Computer Science}, 813:375--409, 2020.

\bibitem{AG2020}
F.~Aschieri and F.A. Genco.
\newblock Par means parallel: multiplicative linear logic proofs as concurrent
  functional programs.
\newblock {\em {PACMPL}}, 4({POPL}):18:1--18:28, 2020.

\bibitem{AschieriZH}
F.~Aschieri and M.~Zorzi.
\newblock On natural deduction in classical first-order logic: Curry--howard
  correspondence, strong normalization and herbrand's theorem.
\newblock {\em Theoretical Computer Science}, 625:125--146, 2016.

\bibitem{Avron91}
A.~Avron.
\newblock Hypersequents, logical consequence and intermediate logics for
  concurrency.
\newblock {\em Annals of Mathematics and Artificial Intelligence},
  4(3):225--248, 1991.

\bibitem{barendregt84}
H.P.\ Barendregt.
\newblock {\em The lambda calculus}.
\newblock North-Holland Amsterdam, 1984.

\bibitem{CairesPfenning}
L.~Caires and F.~Pfenning.
\newblock Session types as intuitionistic linear propositions.
\newblock In {\em {CONCUR} 2010}, pages 222--236, 2010.

\bibitem{Montesi}
M.~Carbone, F.~Montesi, and C.~Sch\"umann.
\newblock Choreographies, logically.
\newblock {\em Distributed Computing}, 31(1):51--67, 2018.

\bibitem{CG2018}
A.~Ciabattoni and F.A. Genco.
\newblock Hypersequents and systems of rules: Embeddings and applications.
\newblock {\em {ACM TOCL}}, 19:1--27, 2018.

\bibitem{Marcello}
M.\ D'Agostino.
\newblock An informational view of classical logic.
\newblock {\em Theor. Comp. Sci.}, 606:79--97, 2015.

\bibitem{MFG13}
M.\ D'Agostino, M.\ Finger, and Dov~M.\ Gabbay.
\newblock Semantics and proof-theory of depth bounded boolean logics.
\newblock {\em Theor. Comp. Sci.}, 480:43--68, 2013.

\bibitem{DanosKrivine}
V.~Danos and J.-L. Krivine.
\newblock Disjunctive tautologies as synchronisation schemes.
\newblock In {\em {CSL} 2000}, pages 292--301, 2000.

\bibitem{codemobil}
A.~Fuggetta, G.P. Picco, and G.~Vigna.
\newblock Understanding code mobility.
\newblock In {\em {IEEE} Transactions on Software Engineering}, volume~24,
  pages 342--361, 1998.

\bibitem{Girard}
J.-Y. Girard, Y.~Lafont, and P.~Taylor.
\newblock {\em Proofs and Types}.
\newblock Cambridge University Press, 1989.

\bibitem{Harper}
R.~Harper.
\newblock Parallelism is not concurrency.
\newblock Available at
  \url{https://existentialtype.wordpress.com/2011/03/17/parallelism-is-not-concurrency},
  2011.

\bibitem{HT}
K.~Honda and M.~Tokoro.
\newblock An object calculus for asynchronous communication.
\newblock In {\em {ECOOP} 1991}, pages 133--147, 1991.

\bibitem{Grace}
T.~Horstmeyer and R.~Loogen.
\newblock Graph-based communication in {E}den.
\newblock {\em Higher-order and symbolic computation}, 26(1):3--28, 2013.

\bibitem{landin}
P.J. Landin.
\newblock The mechanical evaluation of expressions.
\newblock {\em The Computer Journal}, 6(4):308--320, 1964.

\bibitem{loogen2011}
R.~Loogen.
\newblock {E}den -- parallel functional programming with {H}askell.
\newblock In {\em {CEFP} 2011}, pages 142--206, 2011.

\bibitem{Eden}
R.~Loogen, Y.~Ortega-Mall\`en, and R.~Pena.
\newblock Parallel functional programming in {E}den.
\newblock {\em Journal of Functional Programming}, 15(3):431--475, 2005.

\bibitem{L1982}
E.~G.~K. L\'opez-Escobar.
\newblock Implicational logics in natural deduction systems.
\newblock {\em Journal of Symbolic Logic}, 47(1):184--186, 1982.

\bibitem{Milner}
R.~Milner.
\newblock Functions as processes.
\newblock {\em Mathematical Structures in Computer Science}, 2(2):119--141,
  1992.

\bibitem{MCHP}
T.~{Murphy~VII}, K.~Crary, R.~Harper, and F.~Pfenning.
\newblock A symmetric modal lambda calculus for distributed computing.
\newblock In {\em {LICS} 2004}, pages 286--295, 2004.

\bibitem{Parigot92}
M.~Parigot.
\newblock Lambda-mu-calculus: An algorithmic interpretation of classical
  natural deduction.
\newblock In {\em {LPAR} 1992}, pages 190--201, 1992.

\bibitem{Prawitz}
D.~Prawitz.
\newblock Ideas and results in proof theory.
\newblock In {\em Proceedings of the Second Scandinavian Logic Symposium},
  pages 237--309. North-Holland, 1971.

\bibitem{TCP2013}
B.~Toninho, L.~Caires, and F.~Pfenning.
\newblock Higher-order processes, functions, and sessions: A monadic
  integration.
\newblock In {\em {ESOP} 2013}, pages 350--369, 2013.

\bibitem{Wadler2012}
P.\ Wadler.
\newblock Propositions as sessions.
\newblock {\em {ICFP} 2012}, 24:384--418, 2012.

\bibitem{Wadler}
P.~Wadler.
\newblock Propositions as types.
\newblock {\em Communications of the {ACM}}, 58(12):75--84, 2015.

\end{thebibliography}

\newpage 
\appendix
\section{Appendix}
\begin{theorem}[Adequacy Theorem]\label{Adequacy Theorem}
Suppose that $w: A$ is any simply typed $\lambda$-term, with intuitionistic free variables among $x_1: {A_1},\ldots,x_n:{A_n}$. For all terms $t_1, \ldots, t_n$ such that
 for $i=1,\ldots, n, \quad t_i\red A_i$
\[\mbox{we have} \quad w[t_1/x_1\cdots
t_n/x_n]\red A\]
\end{theorem}

\begin{proof}

Notation: for any term $v$ and formula $B$, we denote \[v[t_1/x_1\cdots t_n/x_n]\] with $\substitution{v}$. We proceed by induction on the shape of $w$. 

\begin{enumerate}

\item
If $w= x_{i}: A_{i}$, for some $i$, then
$A=A_i$. So $\substitution{w}=t_i\red
{A_i}={A}$.

\item
If $w= a: A_{i}$, for some $i$ and channel $a$, then
$A=A_i$. By Prop.~\ref{prop:channels}, $\substitution{a}=a\red A$.

\item If $w=ut$, then $ u:
B\rightarrow A$ and $ t: B$. So
$\substitution{w}=\substitution{u}\,\substitution{t}\red {A}$, for
$\substitution{u}\red {B}\rightarrow {A}$ and
$\substitution{t}\red {B}$ by induction hypothesis.

  \item If  $w=\lambda x^{\scriptscriptstyle B}\,
u$, then
$A=B\rightarrow C$. So, $\substitution{w}=\lambda x^{\scriptscriptstyle B}\, \substitution{u}$, since we may assume $x^{\scriptscriptstyle B} \neq x_1, \ldots, x_k$. For  every  $t\red
{B}$, by induction hypothesis on $u$, $\substitution{u}[t/x]\red{C}$.  
Hence, by Proposition \ref{proposition-somecases}, $\lambda x^{\scriptscriptstyle B}\, \substitution{u}\red {B}\rightarrow {C}=
{A}$.

\item The cases of pairs and projections are straightforward.\end{enumerate}\end{proof}


We present an additional example of $\lama$ program. We simulate, in particular, a typical client-server interaction of an online sale.
\begin{example}[\textbf{Buyer and Vendor}]\label{ex:sale} We model the
following transaction: a buyer tells a vendor a product name
$\mathsf{prod}:\mathsf{String}$, the vendor computes the monetary cost
$\mathsf{price}:\mathbb{N}$ of $\mathsf{prod}$ and communicates it to
the buyer, the buyer sends back the credit card number $\mathsf{card}:
\mathsf{String}$ which is used to pay.
We introduce the following functions: $\mathsf{cost} :
\mathsf{String}\rightarrow \mathbb{N} $ with input a product name
$\mathsf{prod}$ and output its cost $\mathsf{price}$;
$\mathsf{pay\_for}: \mathbb{N}\rightarrow \mathsf{String} $ with input
a $\mathsf{price}$ and output a credit card number $\mathsf{card}$;
$\mathsf{use}: \mathsf{String}\rightarrow \mathbb{N}$ that produces
money using as input a credit card number $\mathsf{card}:
{\mathsf{String}} $.
The buyer and the vendor are the contexts
$\mathcal{B}$ and $\mathcal{V}$ of type $\mathsf{Bool}$.
The  communication channel $a$ is typed using the instance $(\mathsf{String} \impl \mathbb{N}) \vel (\mathbb{N} \impl \mathsf{String} ) $ of
the axiom $(A \impl B) \vel ( B \impl A) $.
The program is: \begin{small}\begin{align*} & \pp{a}{ \mathcal{B}[\send{a}(
\mathsf{pay\_for}(\send{a}(\mathsf{prod})))] \parallel
\mathcal{V}[\mathsf{use}( \send{a}( \mathsf{cost}(\get{a}\, 0)))]} \\ &\mapsto
\pp{a}{\mathcal{B}[\send{a}( \mathsf{pay\_for}(\get{a}(\mathsf{prod})))] \parallel \mathcal{V}[\mathsf{use}( \send{a}( \mathsf{cost}( \mathsf{prod} )))
]}
\\ &\mapsto \pp{a}{\mathcal{B}[\send{a}( \mathsf{pay\_for}(\get{a}(\mathsf{prod})))]  \parallel  \mathcal{V}[\mathsf{use}( \send{a}( \mathsf{price}))
]}\\
&\mapsto\pp{a}{\mathcal{B}[\send{a}( \mathsf{pay\_for}(\mathsf{price}))
] \parallel \mathcal{V}[\mathsf{use}( \get{a}( \mathsf{price}))]} \\
&\mapsto \; \pp{a}{\mathcal{B}[\send{a} (\mathsf{card})] \parallel
\mathcal{V}[\mathsf{use}( \get{a}( \mathsf{price}))]} \, \mapsto \,
\pp{a}{\mathcal{B}[\get{a}(
\mathsf{card})] \parallel \mathcal{V}[\mathsf{use}( \mathsf{card})]}
  \end{align*}\end{small}Finally $ \mapsto
\mathcal{V}[\mathsf{use}( \mathsf{card})]$: the buyer has performed its duty and the vendor uses the card number to obtain the due payment.
\end{example}

We present here the unabridged  version of the reduction in Example \ref{ex:floyd-warshall}.
 
\noindent \textbf{Full reduction of the ring Floyd--Warshall algorithm in $\lama$.} In this reduction, we  omit explanations but we indicate the selected process that can receive messages by underlining it.
 
\begin{footnotesize}
  \begin{align*} & \po{a} \big( f(a(f(a(f(a I_1(0)
)))))  \p (a((a((a((a (\lan 0, I_1(0) \ran \pi_1) )
\pi_1))\pi_1))\pi_1) ) \pi_0 \\& \quad  \p (a(\lan 0, I_1(0)  \ran
\pi_1) ) \pi_0  \big) \\ & \p \chosen{\big( f(a(f(a(f(a(f(a \lan
I_2(0) \ran)))))))  \p ( a((a((a((a((a (\lan 0, I_2(0) \ran \pi_1)) \pi_1))\pi_1))\pi_1))\pi_1) ) \pi_0 } \\ &
\quad \chosen{\p (a(f(a( f (a \lan I_2(0) \ran)))) ) \pi_0
 \p ( a((a (\lan 0, I_2(0) \ran \pi_1)) \pi_1) ) \pi_0
} \big) \\ & \p \big( f(a(f(a(f(a(f(a  I_3(0)
)))))))  \p (a((a((a((a (\lan 0, I_3(0) \ran
\pi_1)) \pi_1))\pi_1))\pi_1) )  \pi_0 \\ & \quad \p (a(f(a(
f(a(f (a  I_3(0) )))))) ) \pi_0  \p (a((a((a
(\lan 0, I_3(0) \ran \pi_1)) \pi_1))\pi_1) ) \pi_0 \big)
\pc
  \end{align*}
\end{footnotesize}\begin{footnotesize}
  \begin{align*} \mapsto^* 
  & \po{a}  \big(   f(a(f(a(f(a  I_1(0)))))) \p
      (a((a((a((a   I_1(0) ) \pi_1))\pi_1))\pi_1)) \pi_0 \p
      (a I_1(0)) \pi_0   \big)     \\
    & \p \chosen{\big(   f(a(f(a(f(a(f(a  I_2(0)))))))) \p  (a((a((a((a((a
       I_2(0))
      \pi_1))\pi_1))\pi_1))\pi_1)) \pi_0  }   \\
& \quad\chosen{  \p  (a(f(a( f  (a I_2(0))))) ) \pi_0 \p (a((a   I_2(0)) \pi_1)) \pi_0}  
      \big)     \\ & \p \big(   f(a(f(a(f(a(f(a  I_3(0))))))))
\p (a((a((a((a   I_3(0)) \pi_1))\pi_1))\pi_1)) \pi_0  \\
& \quad   \p (a(f(a( f(a(f  (a  I_3(0)))))))) \pi_0 \p (a((a((a   I_3(0)) \pi_1))\pi_1)) \pi_0 \big) \pc
\end{align*}\begin{align*}
\mapsto^{*}  & \po{a}  \big( f(a(f(a(f(a I_1(0)))))) \p
(a((a((a((a  I_1(0) ) \pi_1))\pi_1))\pi_1)) \pi_0 \p
( a(I_1(0)) ) \pi_0 \big)   \\ 
& \p  \big( f(a(f(a(f(a(f \lan I_2(0), I_1(0) \ran)))))) \p
 (a((a((a((a(\lan  I_2(0), I_1(0) \ran\pi_1))\pi_1))\pi_1))\pi_1) )
  \pi_0    \\
  & \quad \p (a(f(a(f \lan I_2(0), I_1(0) \ran)))) \pi_0 \p (a(\lan
    I_2(0), I_1(0)\ran \pi_{1}  )) \pi_0 \big)
\\ 
& \p\chosen{\big(  f(a(f(a(f(a(f(a I_3(0)))))))) \p (a((a((a((a I_3(0)) \pi_1))\pi_1))\pi_1)) \pi_0}  \\ & \quad \p\chosen{ (a(f(a(
f(a(f (a I_3(0)))))))) \pi_0 \p (a((a((a I_3(0))
\pi_1))\pi_1)) \pi_0} \big) \pc \end{align*}\begin{align*}
\mapsto^* & \po{a}  \big( f(a(f(a(f(a I_1(0)))))) \p
(a((a((a((a  I_1(0) ) \pi_1))\pi_1))\pi_1)) \pi_0 \p
(a(I_1(0)) \big)) \pi_0   \\ 
& \p \big( f(a(f(a(f(a (f \lan I_2(0), I_1(0) \ran)))))) \p
(a((a((a((a I_1(0))\pi_1))\pi_1))\pi_1)) \pi_0 \\&\quad  \p (a(f(a (f \lan
  I_2(0), I_1(0) \ran)))) \pi_0 \p (aI_1(0) ) \pi_0 \big)
\\ 
& \p \chosen{\big(   f(a(f(a(f(a(f(a I_3(0)))))))) \p (a((a((a((a 
I_3(0)) \pi_1))\pi_1))\pi_1)) \pi_0}  \\ & \quad \p \chosen{ (a(f(a(
f(a(f (a I_3(0)))))))) \pi_0 \p (a((a((a  I_3(0))
\pi_1))\pi_1)) \pi_0} \big) \pc \end{align*}\begin{align*}
\mapsto^* & \po{a}  \big(\chosen{ f(a(f(a(f(a I_1(0)))))) \p
(a((a((a((a  I_1(0) ) \pi_1))\pi_1))\pi_1)) \pi_0 \p
(a(I_1(0))) \pi_0 }\big)   \\ 
& \p \big( f(a(f(a(f(a (f \lan I_2(0), I_1(0) \ran)))))) \p
(a((a((a((a I_1(0))\pi_1))\pi_1))\pi_1) ) \pi_0
 \\&\quad 
 \p (a(f(a (f \lan I_2(0), I_1(0) \ran)))) \pi_0 \p (aI_1(0)) \pi_0 \big)  
\\ 
& \p  \big(  f(a(f(a(f(a(f\lan  I_3(0), I_1(0)\ran ))))))) \p
  (a((a((a(\lan 
I_3(0), I_1(0) \ran  \pi_1))\pi_1))\pi_1)) \pi_0  \\ &
                                                                 \quad \p (a(f(a(
f(a(f \lan  I_3(0), I_1(0)\ran )))))) \pi_0 \p (a((a(\lan 
I_3(0) , I_1(0) \ran
\pi_1))\pi_1) ) \pi_0 \big)\pc
\end{align*}
\end{footnotesize}
\begin{footnotesize}
  \begin{align*}
\mapsto^* & \po{a}  \big( \chosen{f(a(f(a(f(a I_1(0)))))) \p
(a((a((a((a  I_1(0) ) \pi_1))\pi_1))\pi_1)) \pi_0 \p
(a(I_1(0))) \pi_0} \big)   \\ 
& \p \big( f(a(f(a(f(a (f \lan I_2(0), I_1(0) \ran)))))) \p
(a((a((a((a I_1(0))\pi_1))\pi_1))\pi_1)) \pi_0  \\&\quad   \p (a(f(a (f \lan I_2(0), I_1(0) \ran)))) \pi_0 \p (aI_1(0)) \pi_0 \big)  
\\
& \p \ \big(  f(a(f(a(f(a (f\lan  I_3(0), I_1(0)\ran )))))) \p
  (a((a((a(I_1(0)))\pi_1))\pi_1)) \pi_0 \\ & \quad \p (a(f(a(
f(a (f\lan  I_3(0), I_1(0)\ran ))))) ) \pi_0 \p (a((a I_1(0))\pi_1)) \pi_0 \big) \pc \end{align*}\begin{align*}  
\mapsto^* & \po{a}  \big( f(a(f(a(f\lan I_1(0) , I_1(0)\ran)))) \p
(a((a((a(\lan  I_1(0) , I_1(0)\ran   \pi_1))\pi_1))\pi_1)) \pi_0
                                                                                                   \\&\quad  \p
\lan I_1(0) , I_1(0)\ran \pi_0 \big)    \\ 
& \p \big(\chosen{ f(a(f(a(f(a (f \lan I_2(0), I_1(0) \ran)))))) \p
 (a((a((a((a I_1(0))\pi_1))\pi_1))\pi_1) ) \pi_0 } \\&\quad \chosen{  \p (a(f(a (f \lan
  I_2(0), I_1(0) \ran))) ) \pi_0 \p (aI_1(0)) \pi_0} \big)  
\\ 
& \p  \big(  f(a(f(a(f(a (f\lan  I_3(0), I_1(0)\ran )))))) \p
  (a((a((a(I_1(0)))\pi_1))\pi_1)) \pi_0 \\ & \quad \p (a(f(a(
f(a (f\lan  I_3(0), I_1(0)\ran )))))) \pi_0 \p (a((a I_1(0))\pi_1)) \pi_0 \big) \pc \end{align*}\begin{align*}  
\mapsto^* & \po{a}  \big( f(a(f(a I_1(1)))) \p
(a((a((a I_1(0))\pi_1))\pi_1)) \pi_0 \p
\lan I_1(0) , I_1(0)\ran \pi_0 \big)    \\ 
& \p \big(\chosen{ f(a(f(a(f(a I_2(1)))))) \p
 (a((a((a((a I_1(0))\pi_1))\pi_1))\pi_1)) \pi_0}  \\&\quad \chosen{ \p (a(f(a I_2(1))) \p aI_1(0)) \pi_0} \big)  
\\ 
& \p  \big(  f(a(f(a(f(a I_3(1)))))) \p
  (a((a((a(I_1(0)))\pi_1))\pi_1) ) \pi_0  \\&\quad  \p (a(f(a(
f(a I_3(1))))) ) \pi_0 \p ( a((a I_1(0))\pi_1)) \pi_0  \big) \pc
\end{align*}
\end{footnotesize}
\begin{footnotesize}
  \begin{align*}
\mapsto^* & \po{a} \big( f(a(f(a I_1(1)))) \p
(a((a((a I_1(0))\pi_1))\pi_1) ) \pi_0 \p
\lan I_1(0) , I_1(0)\ran  \pi_0 \big)    \\ 
& \p  \big( f(a(f(a(f \lan  I_2(1), I_1(0) \ran ))))) \p (
 a((a((a( \lan I_1(0), I_{1}(0)\ran\pi_{1})\pi_1))\pi_1)) \pi_0  \\&\quad  \p (
  a(f\lan  I_2(1) , I_1(0) \ran ))) \pi_0 \p \lan I_1(0), I_1(0) \ran
   \pi_0\big)  
\\ 
& \p  \big(\chosen{  f(a(f(a(f(a I_3(1)))) )) \p
  ( a((a((a(I_1(0)))\pi_1))\pi_1) ) \pi_0 }  \\&\quad  \chosen{\p ( a(f(a(
f(a I_3(1)))))) \pi_0 \p (a((a I_1(0))\pi_1)) \pi_0 } \big) \pc \end{align*}\begin{align*}
 \mapsto^* & \po{a} \big( f(a(f(a I_1(1)))) \p
( a((a((a I_1(0))\pi_1))\pi_1) ) \pi_0\p (\lan I_1(0) , I_1(0)\ran )
             \pi_0 \big)  
\\ 
& \p  \big( f(a(f(a I_2(1)))) \p
 ( a((a((aI_1(0))\pi_1))\pi_1)) \pi_0 \p (a( I_2(1)  )) \pi_0 \p (
  \lan I_1(0), I_1(0) \ran ) \pi_0 \big)    \\ 
& \p \big(\chosen{
 f(a(f(a(f(a I_3(1)))))) \p ( a((a((a(I_1(0)))\pi_1))\pi_1) ) \pi_0}
                                                                              \\&\quad  \chosen{ \p
  (a(f(a( f(a
I_3(1))))) ) \pi_0 \p ( a((a I_1(0))\pi_1)) \pi_0 } \big)\pc
\end{align*}
\end{footnotesize}
\begin{footnotesize}
  \begin{align*}
\mapsto^* & \po{a} \big(\chosen{ f(a(f(a I_1(1)))) \p
( a((a((a I_1(0))\pi_1))\pi_1) ) \pi_0 \p \lan I_1(0) , I_1(0)\ran
            \pi_0 } \big)  
\\ 
& \p  \big( f(a(f(a I_2(1)))) \p
 ( a((a((aI_1(0))\pi_1))\pi_1)) \pi_0  \p ( a \, I_2(1)   ) \pi_0 \p
  \lan I_1(0), I_1(0) \ran \pi_0 \big)   \\ 
& \p   \big(
 f(a(f(a(f \lan  I_3(1), I_2(1)\ran )) ))\p
  ( a((a( \lan I_{3}(1), I_2(1)\ran \pi_{1})\pi_1)) \pi_0   \\
& \quad \p ( a(f(a( f \lan  I_3(1), I_2(1)\ran
  ))) ) \pi_0 \p ( a(\lan I_1(0), I_2(1)\ran \pi_1) ) \pi_0 \big)\pc
\end{align*}
\end{footnotesize}
\begin{footnotesize}
  \begin{align*}
\mapsto^* & \po{a} \big( \chosen{f(a(f(a I_1(1)))) \p (
a((a((a I_1(0))\pi_1))\pi_1) ) \pi_0 \p \lan I_1(0) , I_1(0)\ran \pi_0
            } \big)  
\\ 
& \p  \big( f(a(f(a I_2(1)))) \p
 ( a((a((aI_1(0))\pi_1))\pi_1)) \pi_0 \p ( a I_2(1)  ) \pi_0 \p \lan
  I_1(0), I_1(0) \ran \pi_0  \big)   \\ 
& \p  \big(
 f(a(f(a(f \lan  I_3(1), I_2(1)\ran  ))))\p
  ( a((a I_2(1))\pi_1) ) \pi_0  \\&\quad  \p (a(f(a I(f \lan  I_3(1), I_2(1)\ran
  ))) ) \pi_0 \p ( a I_2(1) ) \pi_0 \big) \pc \end{align*}\begin{align*}
 \mapsto^* & \po{a} \big( f(a(f\lan I_1(1) , I_2(1)\ran ))  \p
( a((a(\lan  I_1(0), I_2(1)\ran \pi_1))\pi_1) ) \pi_0 \p \lan I_1(0) ,
            I_1(0)\ran \pi_0  \big)   
\\ 
& \p  \big(\chosen{ f(a(f(a I_2(1)))) \p
 ( a((a((aI_1(0))\pi_1))\pi_1) ) \pi_0 \p ( a I_2(1)  ) \pi_0 \p \lan
  I_1(0), I_1(0) \ran \pi_0 } \big)   \\ 
& \p  \big(
 f(a(f(a (f \lan  I_3(1), I_2(1)\ran )))) \p
  ( a((a I_2(1))\pi_1)) \pi_0  \\&\quad   \p ( a(f(a (f \lan  I_3(1), I_2(1)\ran
  ))) ) \pi_0 \p ( a I_2(1)) \pi_0 \big) \pc \end{align*}\begin{align*}
\mapsto^* & \po{a} \big( f(a(f\lan I_1(1) , I_2(1)\ran )) \p (
a((a I_2(1))\pi_1) ) \pi_0 \p \lan I_1(0) ,
            I_1(0)\ran \pi_0  \big) 
\\ 
& \p  \big(\chosen{ f(a(f(a I_2(1)))) \p
 ( a((a((aI_1(0))\pi_1))\pi_1) ) \pi_0 \p ( a I_2(1)  ) \pi_0 \p \lan
  I_1(0), I_1(0) \ran \pi_0 } \big)   \\ 
& \p  \big(
 f(a(f(a (f \lan  I_3(1), I_2(1)\ran )))) \p
  ( a((a I_2(1))\pi_1) ) \pi_0  \\&\quad  \p ( a(f(a (f \lan  I_3(1), I_2(1)\ran
  ))) ) \pi_0 \p ( a I_2(1)) \pi_0 \big)\pc
\end{align*}
\end{footnotesize}
\begin{footnotesize}
  \begin{align*}
 \mapsto^* & \po{a}  \big( f(a(f\lan I_1(1) , I_2(1)\ran )) \p
( a((a I_2(1))\pi_1) ) \pi_0 \p \lan I_1(0) ,
            I_1(0)\ran \pi_0  \big)   
\\
& \p   \big( f(a(f\lan I_2(1), I_2(1)\ran )) \p
 ( a((a(\lan I_1(0), I_2(1)\ran \pi_1))\pi_1) ) \pi_0  \\&\quad  \p \lan I_2(1),
  I_2(1)\ran \pi_0 \p \lan I_1(0), I_1(0) \ran \pi_0  \big)    \\ 
& \p  \big(\chosen{ f(a(
f(a (f \lan  I_3(1), I_2(1)\ran )) ))\p
  ( a((a I_2(1))\pi_1) ) \pi_0 }  \\&\quad  \chosen{\p ( a(f(a (f \lan  I_3(1), I_2(1)\ran
  ))) ) \pi_0 \p ( a I_2(1)) \pi_0 } \big) \pc \end{align*}\begin{align*}
\mapsto^* & \po{a}  \big( f(aI_1(2)) \p
( a((a I_2(1))\pi_1) ) \pi_0 \p \lan I_1(0) ,
            I_1(0)\ran \pi_0  \big)   
\\ 
& \p   \big( f(a I_2(2)) \p
 ( a((a I_2(1))\pi_1) ) \pi_0 \p \lan I_2(1), I_2(1)\ran \pi_0  \p
  \lan I_1(0), I_1(0) \ran \pi_0  \big)    \\ 
& \p  \big(\chosen{
 f(a(f(a I_3(2)))) \p
  ( a((a I_2(1))\pi_1) ) \pi_0 \p ( a(f(a I_3(2)))) \pi_0  \p ( a
  I_2(1)) \pi_0 } \big) \pc \end{align*}\begin{align*}
\mapsto^* & \po{a}  \big( \chosen{f(aI_1(2)) \p (
a((a I_2(1))\pi_1) ) \pi_0 \p \lan I_1(0) ,
            I_1(0)\ran \pi_0  }\big)   
\\
& \p \big( f(a I_2(2)) \p
 ( a((a I_2(1))\pi_1) ) \pi_0 \p \lan I_2(1), I_2(1)\ran \pi_0 \p \lan I_1(0),
  I_1(0) \ran \pi_0  \big)   
\\ 
&  \p     \big(
 f(a(f\lan I_3(2), I_2(1)\ran))  \p
  ( a(\lan  I_2(1), I_2(1) \ran\pi_1) ) \pi_0  \\&\quad  \p ( a(f\lan I_3(2), I_2(1)\ran)) \pi_0 \p
  \lan  I_2(1), I_2(1) \ran \pi_0 \big)\pc
\end{align*}
\end{footnotesize}\begin{footnotesize}
  \begin{align*} 
  \mapsto^* & \po{a} \big( \chosen{f(aI_1(2)) \p ( a((a
I_2(1))\pi_1) ) \pi_0 \p \lan I_1(0) , I_1(0)\ran \pi_0} \big)   \\ & \p \big(
f(a I_2(2)) \p 
(  a((a I_2(1))\pi_1)) \pi_0  \p \lan I_2(1), I_2(1)\ran \pi_0\p \lan I_1(0),
I_1(0) \ran \pi_0 \big)   \\ & \p  \big(   f(aI_3(2)) \p (a I_2(1)) \pi_0 \p  (a  I_3(2)) \pi_0 \p \lan I_2(1), I_2(1) \ran \pi_0 \big) \pc \end{align*}\begin{align*}
\mapsto^* & \po{a}   \big( f\lan I_1(2),  I_3(2)\ran \p ( a(\lan
I_2(1),  I_3(2) \ran \pi_1)) \pi_0 \p \lan I_1(0) , I_1(0)\ran \pi_0 \big)  
\\ & \p \big(\chosen{ f(a I_2(2)) \p  (a((a I_2(1))\pi_1)) \pi_0 \p
     \lan I_2(1), I_2(1)\ran  \pi_0 \p
\lan I_1(0), I_1(0) \ran \pi_0 } \big)   \\ & \p \big(  f(aI_3(2)) \p
                                              ( a I_2(1) )  \pi_0 \p (
                                              a I_3(2)) \pi_0 \p \lan
                                              I_2(1), I_2(1) \ran  \pi_0 \big) \pc \end{align*}\begin{align*}
\mapsto^* & \po{a}   \big( f\lan I_1(2),  I_3(2)\ran \p ( a   I_3(2) )
             \pi_0 \p \lan I_1(0) ,
I_1(0)\ran  \pi_0 \big)   \\ & \p \big(\chosen{ f(a I_2(2)) \p  ( a((a I_2(1))\pi_1)) \pi_0  \p \lan
I_2(1), I_2(1)\ran  \pi_0 \p \lan I_1(0), I_1(0) \ran \pi_0} \big)   \\ & \p
\big(   f(aI_3(2)) \p ( a I_2(1) )  \pi_0 \p ( a  I_3(2))  \pi_0 \p \lan
I_2(1), I_2(1) \ran  \pi_0 \big) \pc 
\end{align*}
\end{footnotesize}\begin{footnotesize}
  \begin{align*}
\mapsto^* & \po{a}  \big( f\lan I_1(2),  I_3(2)\ran \p ( a  I_3(2))  \pi_0 \p \lan I_1(0) ,
I_1(0)\ran \pi_0 \big)    \\ & \p  \big( f \lan  I_2(2),  I_3(2)\ran
                         \p  ( a( \lan  I_2(1),  I_3(2)\ran \pi_1) )  \pi_0\p \lan
I_2(1), I_2(1)\ran \pi_0 \p \lan I_1(0), I_1(0) \ran  \pi_0 \big)   \\ & \p
\big(\chosen{  f(aI_3(2)) \p ( a I_2(1) )  \pi_0 \p ( a I_3(2)) \pi_0 \p \lan
I_2(1), I_2(1) \ran \pi_0} \big) \pc \end{align*}\begin{align*}
\mapsto^* & \po{a}  \big(f\lan I_1(2),  I_3(2)\ran \p ( a I_3(2) )
            \pi_0 \p \lan I_1(0) ,
I_1(0)\ran  \pi_0 \big)    \\ & \p  \big(   f \lan  I_2(2),
                                I_3(2)\ran  \p  ( a   I_3(2)) \pi_0 \p \lan
I_2(1), I_2(1)\ran  \pi_0 \p \lan I_1(0), I_1(0) \ran \pi_0 \big)  \\ & \p
\big(\chosen{  f(aI_3(2)) \p ( a I_2(1) )  \pi_0 \p ( a I_3(2)) \pi_0 \p \lan
I_2(1), I_2(1) \ran \pi_0 } \big) \pc 
\end{align*}
\end{footnotesize}
\begin{footnotesize}
  \begin{align*}
\mapsto^* & \po{a} \big(\chosen{ f\lan I_1(2),  I_3(2)\ran \p ( a
            I_3(2))  \pi_0 \p \lan I_1(0) ,
I_1(0)\ran \pi_0} \big) \\ & \p  \big( f \lan  I_2(2),  I_3(2)\ran \p
                             ( a I_3(2) )  \pi_0 \p
\lan I_2(1), I_2(1)\ran \pi_0 \p \lan I_1(0), I_1(0) \ran \pi_0 \big) \\ & \p
 \big( f\lan I_3(2), I_3(2) \ran \p \lan I_2(1), I_3(2) \ran \pi_0  \p
                                                                           \lan I_3(2), I_3(2) \ran  \pi_0 \p \lan I_2(1), I_2(1)
\ran  \pi_0 \big) \pc \end{align*}\begin{align*}
\mapsto^* & \po{a} \big(\chosen{ I_1(3) \p ( a I_3(2))  \pi_0 \p \lan I_1(0) ,
I_1(0)\ran \pi_0 } \big) \\ & \p  \big( I_2(3) \p  ( a I_3(2))  \pi_0 \p
\lan I_2(1), I_2(1)\ran \pi_0 \p \lan I_1(0), I_1(0) \ran  \pi_0\big) \\ & \p
 \big(  I_3(3) \p \lan I_2(1), I_3(2) \ran \pi_0  \p \lan I_3(2), I_3(2) \ran \pi_0 \p \lan I_2(1), I_2(1)
\ran \pi_0  \big) \pc
  \end{align*}
\end{footnotesize}
\[
    \mapsto^* I_1(3) \p I_2(3) \p I_3(3) \]

\end{document}